\documentclass[10pt,a4paper,draft]{article}

\usepackage{amssymb,amsmath}
\usepackage{graphicx,graphics}
\usepackage{mathtools}
\usepackage{bbold}
\usepackage[english]{babel}
\usepackage[utf8]{inputenc}
\usepackage{epsfig,url}
\usepackage{bbm,theorem}
\usepackage{a4wide}
\usepackage{enumerate}
\usepackage{calrsfs}
\DeclareMathAlphabet{\pazocal}{OMS}{zplm}{m}{n}

\newtheorem{theorem}{Theorem}[section]
\newtheorem{definition}[theorem]{Definition}
\newtheorem{corollary}[theorem]{Corollary}

\newtheorem{proposition}[theorem]{Proposition}

\def\w#1{\mathop{:}\nolimits\!#1\!\mathop{:}\nolimits}

{\theorembodyfont{\upshape}

\newtheorem{example}[theorem]{Example}
}
\numberwithin{equation}{section}
\numberwithin{theorem}{section}

\newcommand{\qed}{\hfill$\Box$}
\newenvironment{proof}{\begin{trivlist}\item[]{\em Proof:}\/}{%
\qed\end{trivlist}}

\newcommand{\E}{{\mathbb E}}
\newcommand{\Z}{{\mathbb Z}}
\newcommand{\R}{{\mathbb R}}
\newcommand{\C}{{\mathbb C\hspace{0.05 ex}}}
\newcommand{\N}{{\mathbb N}}
\newcommand{\T}{{\mathbb T}}

\newcommand{\cf}{{\mathbbm 1}}

\newcommand{\ci}{{\rm i}}
\newcommand{\re}{{\rm Re\,}}
\newcommand{\im}{{\rm Im\,}}
\newcommand{\rme}{{\rm e}}
\newcommand{\rmd}{{\rm d}}

\newcommand{\FT}[1]{\widehat{#1}}

\newcommand{\cancelm}[2]{\widehat{#2}^{(#1)}}
\newcommand{\norm}[1]{\Vert #1\Vert}
\newcommand{\defset}[2]{ \left\{ #1\left|\,
 #2\makebox[0cm]{$\displaystyle\phantom{#1}$}\right.\!\right\} }
\newcommand{\set}[1]{\{#1\}}

\newcommand{\mean}[1]{\langle #1\rangle}
\newcommand{\vep}{\varepsilon}
\newcommand{\defem}[1]{{\em #1\/}}
\newcommand{\qand}{\quad\text{and}\quad}

\newcounter{jlisti}

\newcommand{\Cgen}{g_{\mathrm{c}}}
\newcommand{\Wgen}{G_{\mathrm{w}}}
\newcommand{\Mgen}{G_{\mathrm{m}}}
\newcommand{\NPC}{\text{NPC}}
\newcommand{\Weql}{W^{\text{eql}}}

\newcommand{\email}[1]{E-mail: \tt #1}
\newcommand{\emailjani}{\email{jani.lukkarinen@helsinki.fi}}
\newcommand{\emailmatteo}{\email{matteo.marcozzi@helsinki.fi}}
\newcommand{\addressjani}{\em University of Helsinki,
Department of Mathematics and Statistics\\
\em P.O. Box 68,
FI-00014 Helsingin yliopisto, Finland}

\title{Wick polynomials and time-evolution of cumulants}
\author{Jani Lukkarinen\thanks{\emailjani} , Matteo Marcozzi\thanks{\emailmatteo}\\[1em]
$\,^*, \,^\dag$\addressjani}
\date{\today}

\begin{document}
 \selectlanguage{english}
\maketitle
\begin{abstract}
We show how Wick polynomials of random variables
can be defined combinatorially as the unique choice which removes all ``internal 
contractions'' from 
the related cumulant expansions, also in a non-Gaussian case.
We discuss how an expansion in terms of the Wick polynomials can be used for
derivation of a hierarchy of equations for the time-evolution of cumulants.
These methods are then applied to simplify the formal derivation of
the Boltzmann-Peierls equation in the kinetic scaling limit of the discrete nonlinear 
Schr\"{o}dinger equation (DNLS) with suitable random
initial data.  We also present a reformulation of the standard perturbation 
expansion using cumulants which could
simplify the problem of a rigorous derivation of the Boltzmann-Peierls equation 
by separating the analysis of the solutions to the
Boltzmann-Peierls equation from the analysis of the corrections.  This latter 
scheme is general and not tied to the DNLS evolution equations.
\end{abstract}

\maketitle

\tableofcontents

\section{Introduction}

Wick polynomials, also called Wick products, arose first in quantum field theory 
as a way of regularizing products of field operators.
The principal goal there was to replace monomial products by polynomials with state 
dependent coefficients, chosen so as to remove singular terms appearing in the 
associated perturbation expansion.  

The procedure can also be applied in more general probabilistic settings.  The 
following definition is given in Wikipedia
\cite{wick_wiki} and in the Encyclopedia of Mathematics \cite{wick_encMath}.  Consider $n$
(real) random variables $y_j$, $j=1,2,\ldots,n$, on some probability space 
$(\Omega,\mathcal{B},\mu)$ and denote expectation over the probability measure 
$\mu$ by $\mean{\cdot}$.  The Wick polynomial with powers
$k_j\ge 0$, $j=1,2,\ldots,n$, are then defined recursively in the total degree 
$k_1+k_2+\cdots+k_n$ by the following conditions:
\begin{enumerate}
 \item If $k_1=k_2=\cdots=k_n=0$, set $\w{y_1^{k_1}y_2^{k_2}\cdots 
y_n^{k_n}}=1$.
 \item If the total degree is greater than zero, require that 
$\mean{\w{y_1^{k_1}y_2^{k_2}\cdots y_n^{k_n}}}=0$.
 \item  For all $j$, require that the (algebraic) derivatives of the Wick 
polynomials satisfy
 \begin{align}\label{eq:Appell}
 \partial_{y_j} \w{y_1^{k_1}\cdots y_{j}^{k_j}\cdots y_n^{k_n}} = k_j 
\w{y_1^{k_1}\cdots y_{j}^{k_j-1}\cdots y_n^{k_n}}\, .  
 \end{align}
\end{enumerate}
These conditions have a unique solution for which $\w{y_1^{k_1}y_2^{k_2}\cdots 
y_n^{k_n}}$ is a polynomial of total degree $k_1+k_2+\cdots+k_n$ in the variables $y_j$. (The uniqueness is
algebraic, not only almost everywhere as random variables.  
That is, the conditions fix all coefficients of the polynomials.  This can be seen by induction 
in the order $|I|$: the requirement in item 3 fixes all new coefficients apart from the constant, which is then
fixed by the vanishing of the expectation value in item 2.)
The coefficients are polynomials of expectations of the 
random variables $y_j$, and hence depend on the measure $\mu$.  The first order 
polynomial is obtained by simply centering the variable, 
$\w{y_1}=y_1-\mean{y_1}$, but already at second order more complex structures 
appear,
$\w{y_1y_2}=y_1 y_2-\mean{y_1}y_2 -\mean{y_2}y_1-\mean{y_1y_2}+2 
\mean{y_1}\mean{y_2}$.

If the random variables have joint exponential moments, i.e., if there is 
$\beta>0$ such that 
$\mean{\rme^{\beta \sum_j |y_j|}}<\infty$,
the Wick polynomials can also be obtained by differentiating a fairly simple 
generating function.  It can then be defined for  $\lambda\in \R^n$, such that 
$|\lambda_j|<\beta$ for all $j$, by
\begin{align}
\Wgen(\lambda; y_1, \ldots, y_n)=
 \frac{\exp \left( \sum_{i=1}^n \lambda_i y_i \right)}{\mean{\exp \left( 
\sum_{i=1}^n \lambda_i y_i  \right)}} \, ,
\end{align}
and then for all $k_j\ge 0$, $j=1,2,\ldots,n$, 
\begin{align}
 \w{y_1^{k_1}y_2^{k_2}\cdots y_n^{k_n}} = \left.\partial^{k_1}_{\lambda_1} 
\cdots \partial^{k_n}_{\lambda_n}  \Wgen(\lambda;y_1, \ldots, y_n) 
\right\vert_{\lambda=0} \, .
\end{align}
The generating function $\Wgen(\lambda; y_1, \ldots, y_n)$ is also called ``Wick 
exponential'' and often denoted by 
``$\w{\exp \left( \sum_{i=1}^n \lambda_i y_i  \right)}$''.  For a derivation and 
basic properties of such Wick polynomials, see \cite{GiSu}.

The Wick polynomials become particularly simple to use if the joint measure of
$y$ is Gaussian.  
Defining the covariance matrix by $C_{j'j} := \mathrm{Cov}(y_{j'},y_j)$, a 
Gaussian measure has
$\mean{\exp \left(\lambda\cdot y\right)}=\rme^{\lambda\cdot 
\mean{y}+\lambda\cdot C\lambda/2}$.  Therefore,
the generating function of the Wick polynomials then reads simply
$\Wgen(\lambda; y)=\exp\!\left[ \lambda\cdot (y-\mean{y})-\lambda\cdot 
C\lambda/2\right]$, and the resulting Wick polynomials are 
closely related to Hermite polynomials.  This is the setting encountered in 
the original problem of renormalization of quantum field theories (the 
``unperturbed measures'' concern free fields and hence are Gaussian).

In the Gaussian case, one can also identify the Wick polynomials as arising from 
an orthogonalization procedure.  
Wiener chaos expansion and Malliavin calculus used for stochastic differential 
equations can be viewed as applications of 
such orthogonal projection techniques. \cite{peccati_taqqu2011}

In the non-Gaussian case, there are far fewer examples of applications of Wick 
polynomial techniques.  The computations become then
more involved.  For instance, there is no explicit formula for the generating 
function unless the inverse of the moment generating function happens to be 
known explicitly.  In addition, then
the polynomials typically no longer form an orthogonal set in $L^2(\mu)$.

The goals of the present contribution are two-fold.  In the first part, we show 
that
Wick polynomials have a natural combinatorial definition, closely connected to 
cumulants and the related cluster expansions 
of correlation functions.  We also rederive their main properties without 
resorting to the generating function, hence without assuming Gaussianity or the 
existence of exponential moments.

In the second part, we show how Wick polynomial expansions may be used in 
the analysis of
stochastic processes.
In particular, the goal there is to apply the expansion to study the 
time-evolution of the cumulants, i.e., of the connected correlation functions, 
of the process.  We will explain there why often it is cumulants, and not 
moments, which should be used as dynamical variables.
For simplicity, we consider here only processes whose dynamics are deterministic 
and given by a differential equation, such as Hamiltonian evolution in classical 
particle systems.  The randomness enters via the initial state.
However, generalization to Markovian stochastic dynamics should be 
straightforward, for instance,  
if the generator of the process maps polynomials to polynomials.

In the general setup, the best one can hope for are recursion relations leading 
to an infinite hierarchy
of equations connecting the evolution of the cumulants.  We explain in Section 
\ref{sec:cwdyn} what immediate
constructions are available for hierarchical study of the evolution of cumulants 
and Wick polynomials.

We give more explicit applications in section \ref{sec:DNLS} where we study the 
evolution on a lattice of particles following the discrete nonlinear Schr\"{o}dinger 
(DNLS) equation with random initial data.
In particular, our goal is to show how the Wick polynomial expansion of the 
dynamics greatly simplifies the (still only formal) derivation of the related 
Boltzmann-Peierls equation.  
This case is one of the few examples of nonlinear Hamiltonian evolution where a 
rigorous analysis of the related perturbation expansion has been possible so 
far: in \cite{LuSp11} it was proven that if the initial measure is a stationary 
Gibbs measure, then the time-correlations of the field follow an evolution 
equation derived using a perturbation expansion analogous to the one needed for 
the Boltzmann-Peierls equation.

An ultimate goal of the present reformulation of the evolution problem would be 
to complete the rigorous derivation,
and hence give a region of validity, of the Boltzmann-Peierls equation.  
We show how the Wick polynomial expansion could help in this goal by separating 
the problem of solving the effect of the 
Boltzmann-Peierls evolution from the estimation of the corrections arising from 
the wave nature of the microscopic evolution, such as constructive interference. 
 For the DNLS evolution the Wick polynomial expansion coincides with what was 
called ``pair truncation'' in \cite{LuSp11}.  In fact, the present work arose 
from an attempt to generalize this construction to other polynomial potentials, 
which we later realized to coincide with Wick polynomial expansions.

It should already be apparent from the above example that in order to use the 
Wick polynomials some care is 
needed in the choice of notations to avoid being overcome by lengthy formulae 
and intractable combinatorial estimates.
We begin by explaining our choices in detail in Section \ref{sec:notations}.  
The first part containing the combinatorial definition and properties of Wick 
polynomials is given in Section \ref{sec:Wick}.  The second part discussing the 
use of Wick polynomial
expansions for the study of evolution of cumulants begins in Section 
\ref{sec:cwdyn}.  
We conclude it with the specific application to DNLS dynamics in Section 
\ref{sec:DNLS}.  
Some comments and possible further directions are discussed in Section 
\ref{sec:discussion}.

\subsection*{Acknowledgements}

We thank Giovanni Gallavotti, Dario Gasbarra, Antti Kupiainen, Peng Mei, Alessia Nota,
Sergio Simonella, and Herbert Spohn for useful discussions on the topic.
The research of J.\ Lukkarinen and M.\ Marcozzi has been supported by the 
Academy of Finland
via the Centre of Excellence in Analysis and Dynamics Research (project 271983) 
and from an Academy Project (project 258302), and also benefited from 
the support of the project EDNHS ANR-14-CE25-0011 of the French National 
Research Agency (ANR).

\section{Setup and notations}\label{sec:notations}

We consider a collection $y_j$, $j\in J$ where $J$ is some fixed nonempty index 
set, of real or complex random variables
on some probability space $(\Omega,\mathcal{B},\mu)$.
If $y_j$ are complex, we assume that the collection is closed under complex 
conjugation, i.e., that to every $j$ there is 
$j'\in J$ such that $y_{j'}= y_j^*$.

Expectation over the probability measure $\mu$ will be denoted by $\E$ or 
$\mean{\cdot}$.  In case the underlying measure needs to be identified, we 
denote the expectation by $\E_\mu$ or $\mean{\cdot}_\mu$.
We use sequences of indices, $I=(i_1,i_2,\ldots,i_n)\in J^n$,
to label monomials of the above random variables, with the following shorthand 
notation
\begin{align}\label{eq:defIpower}
 y^I = y_{i_1} y_{i_2}\cdots y_{i_n} = \prod_{k=1}^n y_{i_k} \, .
\end{align}
We also set $y^\emptyset := 1$ if $I$ is the empty sequence.
Since all $y_j$ commute with each other, we have $y^I=y^{I'}$ for any two 
sequences $I$, $I'$
which differ by a permutation.

We will need to operate not only with such sequences but also with their 
subsequences and ``partitions''.
This will be done by choosing a distinct label for each member of the sequence 
and collecting these into a set.
How the labelling is done is not important, as long as one takes care when 
combining two ``labelled'' sets.
We rely here on the following standard conventions: any sequence $(i_k)$ can be 
uniquely identified with the function $k\mapsto i_k$
which itself is uniquely determined by its graph, the subset 
$\defset{(k,i_k)}{k=1,2,\ldots,n}$ of $\N\times J$.
We consider subsequences to be subsets of the graph of the sequence.  Partitions 
of the sequence then correspond to partitions
of its graph which can be understood as partitions into nonempty subsequences.

Mathematically, this leads to the following structure.  Finite (sub)sequences of 
indices are now uniquely labelled by
the collection $\mathcal{I}$, which consists of those finite subsets $A\subset 
\N\times J$ with the
property that if $(n,j), (n',j')\in A$ and $(n,j)\ne (n',j')$
then $n\ne n'$.  We also allow the sequence to be empty which is identified with 
$\emptyset\in \mathcal{I}$.
For nonempty sets, the natural number in the first component serves as a 
distinct label for each member in $A$.
In addition, we can use the order of the natural numbers to collapse any $A\in 
\mathcal{I}$ back to a sequence $\check{A}$ in $J$:
Given $A\in \mathcal{I}$ with $n>0$ elements,
there is a unique bijection $g:\set{1,2,\ldots,n}\to A$ such that its first 
component is
increasing, $g(k)_1<g(k')_1$ for all $k<k'$.  Using this $g$, we define 
$\check{A}_k:=g(k)_2\in J$ for $k=1,2,\ldots,n$.

To each finite sequence $I=(i_k)$ of $n$ elements in $J$,
we assign $\tilde{I} := \defset{(k,i_k)}{k=1,2,\ldots , n}$ as the set of 
labels.
Obviously, then $\tilde{I}$ and any of its subsets belong to $\mathcal{I}$.
The following list summarizes some basic notations and definitions which will be 
used later without further remark.
\begin{enumerate}
  \item If $I$ is a sequence, and a set is needed by the notation, the set is 
chosen to be $\tilde{I}$.
  For instance, the notation ``$A\subset I$'' means $A\subset \tilde{I}$.
 \item The notation $\pazocal{P}(E)$ denotes the collection of partitions of the 
a set $E\in \mathcal{I}$.
If $I$ is a sequence, $\pazocal{P}(I) := \pazocal{P}(\tilde{I})$.
 \item If $A\in \mathcal{I}$ and it is used in a place of a sequence, the 
formula always refers to the collapsed sequence $\check{A}$
 obtained via the increasing bijection $g$ above.  For instance, then $y^A := 
y^{\check{A}} = \prod_{k=1}^n y_{g(k)_2}$.
 (Note that if $I$ is a sequence, then $y^{\tilde{I}}=y^I$ in agreement with 
(\ref{eq:defIpower}).)
 \item If $A\in \mathcal{I}$, we denote the corresponding sequence of random 
variables by $y_A:=(y_{\check{A}_k})_{k=1}^{|A|}$.
 \item If $m\in \N$ and $A\in \mathcal{I}$, the notation $\cancelm{m}{A}$ refers 
to a set where any element with label $m$ is cancelled,
 i.e., $\cancelm{m}{A} := \defset{(k,i_k)\in A}{k\ne m}$.  Note that it is 
possible that $\cancelm{m}{A}=A$.
 If $I$ is a sequence and $m\le |I|$, $\cancelm{m}{I}$ corresponds to a sequence 
which is obtained by removing the $m$:th member from $I$.
 \item If $I,I'$ are two sequences, they can be merged into a new sequence 
$(i_1,\ldots,i_{|I|},i'_1,\ldots,i'_{|I'|})$
 which we denote by $I+I'$.  If $A,B\in \mathcal{I}$, we take $A+B 
:=\check{A}+\check{B}$.  For a merged
 sequence, the notation ``$I\subset I+I'$'' always refers to the collection of 
the labels of the first $|I|$ members and
 analogously ``$I'\subset I+I'$'' refers to the collection of the last $|I'|$ 
members.  The merge operation is clearly associative, and
 we hence drop parentheses when it is applied iteratively; for instance, 
$I+I'+I''$ is a sequence of length $|I|+|I'|+|I''|$.
\item To avoid separate treatment of expressions involving empty sets and conditions, we 
employ here the following standard conventions:
if the condition $P$ is false, we define
\begin{align}\label{eq:emptyconv}
 \sum_{P} \left(\,\cdots\right) := 0,\quad
 \prod_{P} \left(\,\cdots\right) := 1,\quad\text{and set also}\quad
 \pazocal{P}(\emptyset) := \{\emptyset\}\, .
\end{align}
\end{enumerate}

Similarly to the moments, to any $I\in \mathcal{I}$ we denote the corresponding 
cumulant by one of the following alternative notations
\begin{align}
 & \kappa[y_I] = \kappa_\mu[y_I] = \E[y_{i_1}; y_{i_2}; \cdots; y_{i_n} ] =
 \kappa(y_{i_1}, y_{i_2}, \cdots, y_{i_n})\, .
\end{align}
The corresponding Wick polynomial is denoted by
\begin{align}
 & \w{y_{i_1} y_{i_2} \cdots y_{i_n}}  = \w{y^I} = \w{y^I}_\mu \, .
\end{align}
Note that this notation is slightly formal, since the result is not a function 
only of the power $y^I$ but depends on all subpowers,
$y^A$, $A\subset I$, as well.  It also requires that one carefully defines which 
random variables are being ``Wick contracted''.
We will use parentheses for this purpose, if necessary.  For instance, 
``$\w{\,(y^I)}$'' means $y^I-\E[y^I]$ which usually differs from $\w{y^I}$.

As an application of the above definitions, let us point out that the earlier 
defining Wick polynomial condition (\ref{eq:Appell}) is equivalent to the
requirement that for every nonempty sequence $I$ and any $j\in J$ we should have
\begin{align}\label{eq:Appell2}
 \partial_{y_j} \w{y^I} = \sum_{k=1}^{|I|} \cf(i_k=j) \w{y^{\cancelm{k}{I}}}\, .
\end{align}
Here, and in the following, $\cf$
denotes the generic characteristic function: $\cf(P)=1$ if the condition $P$ is
true, and otherwise $\cf(P)=0$.

We recall that, if the random variables $y_j$, $j=1,2,\ldots,n$, have joint exponential moments, 
then moments, 
cumulants and Wick polynomials can be generated by differentiation of their 
respective generating functions which are
\begin{align}\label{eq:defallgenf}
\Mgen(\lambda) := \E[\rme^{\lambda\cdot x}]\, , \quad
\Cgen(\lambda) := \ln \Mgen(\lambda) \qand
\Wgen(\lambda; y) := \frac{\rme^{\lambda\cdot y}}{\E[\rme^{\lambda\cdot x}]} =
\rme^{\lambda\cdot y-\Cgen(\lambda)}
\, .
\end{align}
Here $\lambda\cdot x:=\sum_{i=1}^n \lambda_i x_i$ (for the sake of clarity we 
have denoted the integrated random variable by ``$x$'' instead of ``$y$'') and 
the ``generation'' happens by evaluation of the $I$:th derivative at zero, 
i.e., 
\begin{align}
\E[y^I] = \partial^I_\lambda\Mgen(0)\, , \quad
\kappa[y_I]= \partial^I_\lambda\Cgen(0) \qand
\w{y^I} = \partial^I_\lambda\Wgen(0;y)
\, ,
\end{align}
where ``$\partial^I_\lambda$'' is a shorthand notation for 
$\partial_{\lambda_{i_1}}\partial_{\lambda_{i_2}}\cdots 
\partial_{\lambda_{i_n}}$.

As a side remark, let us also recall that it is possible to replace the above definitions of 
generating functions by parametrizations which
do not require the existence of any moments and hence work for arbitrary Borel 
probability measures $\mu$.
If all $y_j$ are real, then
replacing the exponential $\rme^{\lambda\cdot y}$ by $\rme^{\ci \lambda\cdot y}$ 
yields an $L^1(\mu)$ function for all 
$\lambda\in \R^n$.  If $y\in \C^n$, the same is achieved by using $\rme^{\ci 
(\lambda^*\cdot y+\lambda\cdot y^*)/2}$
and $\lambda\in \C^n$: in this case, differentiation with respect to $\re 
\lambda_j$ generates ``$\ci \,\re y_j$'' and with
respect to $\im \lambda_j$ generates ``$\ci \,\im y_j$''.  
Naturally, without absolute integrability of the moments it is not 
guaranteed that any of the derivative exist.  
However, it might nevertheless be useful to inspect the time evolution of the 
generating function, in particular, 
if the time evolution is regularizing and improves the integrability of the 
moments.

\section{Combinatorial definition and properties of the Wick polynomials}
\label{sec:Wick}

Let us first recall the ``moments-to-cumulants'' formula which holds for any 
$I\in \mathcal{I}$ as long as all moments $y^A$, $A\subset I$, belong to $L^1(\mu)$:
\begin{align}\label{eq:mtc}
 \mathbb{E}[y^I] = \sum_{\pi \in\pazocal{P}(I)} \prod_{A\in \pi} \kappa[y_A]\, ,
\end{align}
where $\pazocal{P}(I)$ denotes the collection of partitions of the set $I$. (Or 
to be precise, of $\tilde{I}$.  Here it is
important to assign a distinct label to each random variable in the power $y^I$ 
to get the combinatorics correctly.)
For a partition $\pi  \in\pazocal{P}(I)$, let us call the subsets $A\in \pi$ 
\defem{clusters} or \defem{blocks}.
Let us also recall that the cumulants are multilinear, i.e., they are 
separately 
linear in each of the variables $y_j$, $j\in I$.  
These results are discussed, for instance, in  \cite{Brillinger1975} and \cite{Shiryaev1995} and also 
briefly in Appendix \ref{app:cumulants} here.

Let us point out that by the conventions adopted here, (\ref{eq:mtc}) is indeed valid also for the empty
sequence $I=\emptyset$.  Then the sum over partitions is not empty since it contains $\pi=\emptyset$.  However, the corresponding 
product is empty since there is no $A$ with $A\in \pi$.  Therefore, the right hand side of (\ref{eq:mtc}) evaluates to one which agrees with
our definition of $\mathbb{E}[y^\emptyset]$.

We next show that it is possible to choose a subset of the indices and remove 
all its ``internal clusters''
from the moments-to-cumulants formula by replacing the corresponding power with 
a polynomial of the same order.
This will be achieved by using the following recursive definition.
\begin{definition}\label{comb_definition}
 Suppose that $I_0\in \mathcal{I}$ is such that $\E[|y^I|]<\infty$ for all 
$I\subset I_0$.  We define polynomials
 $\pazocal{W}[y^I] := \sum_{E\subset I} c_E[y^I]\, y^E$ for $I\subset I_0$ 
inductively in $|I|$ using the following rule:
 set $\pazocal{W}[y^{\emptyset}]:=1$, and for $I\ne \emptyset$ use
\begin{align}\label{eq:defTJ}
  \pazocal{W}[y^{I}] := y^I - \sum_{\emptyset\ne E\subset I} \E[y^E]\,
  \pazocal{W}[y^{I\setminus E}]\, .
\end{align}
\end{definition}
The definition makes sense since 
the $\pazocal{W}$-terms on the right hand side all have an order lower than 
$|I|$. It also implies that indeed each $\pazocal{W}[y^I]$
is a polynomial of order $|I|$, with only the term $y^I$ being of the highest 
order.
It is also straightforward to prove by induction that the coefficients 
$c_E[y^I]$ can be chosen so that they only depend on $\E[y^A]$ with $A\subset 
I$.
Therefore, the definition of $\pazocal{W}[y^{I}]$ is independent of $I_0$ in the 
following sense: if $I_0,I'_0\in 
\mathcal{I}$ are such that $\E[|y^I|]<\infty$ for every $I\subset I_0$ and 
$I\subset I'_0$, 
then for all $I\subset I_0\cap I'_0$ we have $\pazocal{W}[y^I;I_0] = 
\pazocal{W}[y^I;I'_0]$.
In Appendix \ref{app:cumulants} we explain how cumulants can also be defined 
via a similar recursive construction.

The following theorem shows that these polynomials indeed have the promised 
truncated
moments-to-cumulants expansion. We also see that the polynomials are essentially 
uniquely defined by this property.
What is perhaps surprising is that the coefficients of the polynomial can be 
chosen 
depending  only on the moments of its constituent random variables.  This implies 
that the same polynomial can be used for many different probability distributions, as long as the
marginal distributions for the constituent random variables are the same.
\begin{theorem}\label{th:W_prop}
Assume that the measure $\mu$ has all moments of order $N$, i.e.,  suppose that 
$\E[|y^I|]<\infty$ for all $I\in \mathcal{I}$ with 
$|I|\le N$.  
Use Definition \ref{comb_definition} to define $\pazocal{W}[y^I]$ for every such 
$I$.

Then replacing $y^I$ by $\pazocal{W}[y^I]$ removes all terms with clusters 
internal to $I$:
the following truncated moments-to-cumulants formula holds for every 
$I'\in \mathcal{I}$ with $|I'|+|I|\le N$
\begin{align}\label{eq:maintJ}
 \mathbb{E}\!\left[\pazocal{W}[y^I]y^{I'} \right] = \sum_{\pi 
\in\pazocal{P}(I+I')} 
 \cf(A\cap I' \ne \emptyset\text{ for all }A\in\pi) \prod_{A\in \pi} 
\kappa[y_A]\, .
\end{align}
In particular, $\mathbb{E}\!\left[\pazocal{W}[y^I] \right]=0$ if $I\ne 
\emptyset$.  

In addition, if $I\in \mathcal{I}$ with $|I|\le N/2$ and 
$\pazocal{W}'$ is a polynomial of order at most $|I|$
such that (\ref{eq:maintJ}) holds 
for all $I'$ with $|I'|\le N-|I|$, then $\pazocal{W}'$
is $\mu$-almost surely equal to $\pazocal{W}[y^I]$.
\end{theorem}
\begin{corollary}
Assume that $\E[|y^I|]<\infty$ for all $I\in \mathcal{I}$.
Then $\pazocal{W}[y^I]$ are $\mu$-almost surely
unique polynomials of order $|I|$ such that (\ref{eq:maintJ}) holds for every 
$I'\in \mathcal{I}$.
\end{corollary}
\begin{proof}  We make an induction in $|I|$.
By (\ref{eq:mtc}), the claim is true for $|I|=0$ since then $I=\emptyset$ and 
thus
$\pazocal{W}[y^I]=1$.  
 
Assume then that $I\ne \emptyset$ and that the claim is true for sets of size 
less than $|I|$.  Consider an arbitrary 
$I'\in \mathcal{I}$ such that $|I'|+|I|\le N$.
For $E\subset I$, denote $E^c:=(I+ I')\setminus E$.
Given a partition $ \pi$ of $I+I' $, we can define
$$
\pi_1 :=\left\lbrace A_1 \in \pi \, | \, A_1 \cap I' \neq \emptyset
 \right\rbrace 
$$
and
$$ 
\pi_0 := \pi \backslash \pi_1 \, .
$$
Then $E:=\cup \pi_0 \subset I$ and $\pi_0\in \pazocal{P}(E)$, $\pi_1\in 
\pazocal{P}(E^c)$ (also whenever $E$ or $E^c$ happens to be empty).
Once $ \pi$ is fixed, the decomposition $\pi = \pi_0 \cup \pi_1 $ is unique and 
we thus find that $1=\sum_{E\subset I} \sum_{\pi_0 \in\pazocal{P}(E)}
\sum_{\pi_1 \in\pazocal{P}(E^c)} \cf(\pi=\pi_0\cup \pi_1) \prod_{A_1\in 
\pi_1}\cf(A_1\cap I' \ne \emptyset)$.
Using this in the standard moments-to-cumulants formula shows that
\begin{align*}
& \mathbb{E}\!\left[y^I y^{I'} \right] = \mathbb{E}\!\left[y^{I+I'} \right] =
 \sum_{\pi \in\pazocal{P}(I+I')} \prod_{A\in \pi} \kappa[y_A]
 \nonumber \\
& \quad = 
\sum_{\pi \in\pazocal{P}(I+I')} \sum_{E\subset I} \sum_{\pi_0 \in\pazocal{P}(E)}
\sum_{\pi_1 \in\pazocal{P}(E^c)}
\cf(\pi=\pi_0\cup \pi_1) 
\prod_{A_0\in \pi_0} \kappa[y_{A_0}] \prod_{A_1\in \pi_1} 
\left(\kappa[y_{A_1}]\cf(A_1\cap I' \ne \emptyset)\right)  \nonumber  \\
& \quad = \sum_{E\subset I} \sum_{\pi_0 \in\pazocal{P}(E)}
\sum_{\pi_1 \in\pazocal{P}(E^c)}
\prod_{A_0\in \pi_0} \kappa[y_{A_0}] \prod_{A_1\in \pi_1} 
\left(\kappa[y_{A_1}]\cf(A_1\cap I' \ne \emptyset)\right)  \nonumber  \\
& \quad = \sum_{\pi \in\pazocal{P}(I+I')}\cf(A\cap I' \ne \emptyset\ \forall 
A\in\pi)\prod_{A\in \pi} \kappa[y_{A}] \nonumber \\
& \qquad + 
\sum_{\emptyset \neq E\subset I}\sum_{\pi_0 \in\pazocal{P}(E)} \prod_{A_0\in 
\pi_0}\kappa[y_{A_0}]
\sum_{\pi_1 \in\pazocal{P}(E^c)} \cf(A\cap I' \ne \emptyset\ \forall A\in\pi_1) 
\prod_{A_1\in \pi_1}\kappa[y_{A_1}]
 \nonumber \\
& \quad =\sum_{\pi \in\pazocal{P}(I+I')}\cf(A\cap I' \ne \emptyset\ \forall 
A\in\pi)\prod_{A\in \pi} \kappa[y_{A}] +\sum_{\emptyset\ne E\subset I}\E[y^E]\,
\mathbb{E}\!\left[\pazocal{W}[y^{I\setminus E}]y^{I'} \right] \, . 
\end{align*}
where in the last step we used the moments-to-cumulants formula and the 
induction hypothesis (note that $E^c$ collapses to the sequence
$(I\setminus E)+I'$).
Hence, by the definition (\ref{eq:defTJ}) equation (\ref{eq:maintJ}) holds for 
this $I$.    
This completes the induction step and shows that (\ref{eq:maintJ}) is valid for 
all $I,I'$ with $|I|+|I'|\le N$. 
If $I\ne \emptyset$ and $I'=\emptyset$, we have $I+I'\ne \emptyset$ so that
for any $\pi \in\pazocal{P}(I+I')$ there is some $A\in \pi$ and then obviously
$A\cap I' = \emptyset$. Thus (\ref{eq:maintJ}) implies that
$\mathbb{E}\!\left[\pazocal{W}[y^I] \right]=0$ for $I\ne \emptyset$.

To prove uniqueness, suppose that $I\in \mathcal{I}$ with $|I|\le N/2$ and 
$\pazocal{W}'$ is a polynomial of order at most $|I|$ such that (\ref{eq:maintJ}) holds 
for all $I'$ with $|I'|\le N-|I|$.
Then $P_I := \pazocal{W}'-\pazocal{W}[y^{I}]$ is a polynomial of order 
at most $|I|$ and $\mathbb{E}[P_I y^{I'}] = 0$ for all $I'$ with $|I'|\le N/2$.  Since the collection of random 
variables is assumed to be closed under complex
conjugation, this implies that also $\mathbb{E}[P_I (y^{I'})^* ] = 
0$ whenever $|I'|\le N/2$.
Thus we can take a linear combination of such equations and 
conclude that 
$\mathbb{E}\!\left[|P_I|^2 \right] = 0$.  This implies that $P_I=0$ almost 
surely, i.e., that 
$\pazocal{W}'=\pazocal{W}[y^{I}]$ almost surely.

This concludes the proof of the Theorem.  The Corollary is then an immediate consequence.
\end{proof}

In fact, the polynomials given by Definition \ref{comb_definition} are equal 
to the standard Wick polynomials.
\begin{proposition}\label{th:WisWick}
Suppose that $I_0\in \mathcal{I}$ is such that $\E[|y^I|]<\infty$ for all 
$I\subset I_0$.
Then $\pazocal{W}[y^I]=\w{y^I}$ for every $I\subset I_0$.
\end{proposition}
\begin{proof}
If $I=\emptyset$, we have $\pazocal{W}[y^I]=1=\w{y^I}$, and else by Theorem 
\ref{th:W_prop} we have $\E[\pazocal{W}[y^I]]=0$.
Therefore, to prove $\pazocal{W}[y^I]=\w{y^I}$ it suffices to check that 
(\ref{eq:Appell2}) holds when the Wick polynomials are replaced by 
$\pazocal{W}$-polynomials.
We do this by induction over $|I|$.  Firstly, if $I=\emptyset$, we have 
$\partial_{y_j} \pazocal{W}[y^I]=0$, as required.
Assume then that $I\ne \emptyset$ and that the claim is true for sets of size 
less than $|I|$.
For every $j\in J$ and nonempty $E\subset I$, the induction assumption implies 
that
\begin{align}\label{eq:WyIeqs}
 \partial_{y_j} \pazocal{W}[y^{I\setminus E}] =
 \sum_{k=1}^{|I|} \cf(i_k=j)\, \cf((k,i_k)\not\in E)\, 
\pazocal{W}[y^{\cancelm{k}{I}\setminus E}]\, .
\end{align}
(The second characteristic function allows keeping the labeling inherited from 
$I$ by adding zero terms into the sum for the ``missing'' labels.)
Since $\pazocal{W}[y^I]$ satisfies (\ref{eq:defTJ}), we thus find that its 
algebraic derivatives satisfy an equality
\begin{align}
  & \partial_{y_j} \pazocal{W}[y^{I}] =
 \sum_{k=1}^{|I|} \cf(i_k=j)\biggl[ y^{\cancelm{k}{I}} -
 \sum_{\emptyset\ne E\subset I} \E[y^E]
 \cf((k,i_k)\not\in E)\, \pazocal{W}[y^{\cancelm{k}{I}\setminus E}] \biggr]
 \nonumber \\ & \quad
 =
 \sum_{k=1}^{|I|} \cf(i_k=j)\biggl[ y^{I'} -
 \sum_{\emptyset\ne E\subset I'} \E[y^E] \pazocal{W}[y^{I'\setminus E}] 
\biggr]_{I'=\cancelm{k}{I}}
 \nonumber \\ & \quad
 =
 \sum_{k=1}^{|I|} \cf(i_k=j) \pazocal{W}[y^{\cancelm{k}{I}}] 
 \, ,
\end{align}
where in the last equality we have applied (\ref{eq:defTJ}).  This proves that 
also $\pazocal{W}[y^{I}]$ satisfies (\ref{eq:Appell2})
and hence completes the induction step.
Therefore, the polynomials $\pazocal{W}[y^I]$ satisfy the defining properties of 
Wick polynomials and thus $\pazocal{W}[y^I]=\w{y^I}$.
\end{proof}

\subsection{Basic properties of the Wick polynomials}
\label{sec:Wickprop}

In this section, we assume that there is $I_0\in \mathcal{I}$ is such that
$\E[|y^I|]<\infty$ for all $I\subset I_0$.
This guarantees the existence of the Wick polynomials $\pazocal{W}[y^I]$ for all $I\subset I_0$, 
and allows using the results from the previous section.  In particular, 
by Theorem \ref{th:WisWick} these are equal to the standard Wick polynomials and from now on we will use the 
standard notation $\w{y^I}$ for them.

The next Proposition collects some of the most important properties of Wick polynomials.
\begin{proposition}
The following statements hold for any $I\subset I_0$:
\begin{enumerate}
\item
\begin{align}
{y^I} = \sum_{U \subset I} \w{y^U} \E[y^{I \backslash U}] =
\sum_{U \subset I} \w{y^U} \sum_{\pi \in \pazocal{P}(I \backslash U)} \prod_{A 
\in \pi} \kappa[y_A] \, .
\label{inv_wick}
\end{align}
\item Wick polynomials are permutation invariant: if $I'$ is a permutation of 
$I$, then $\w{y^{I'}}=\w{y^I}$.
\item
\begin{align}
\w{y^I} = \sum_{U \subset I} y^U \sum_{\pi \in \pazocal{P}(I \backslash 
U)}(-1)^{|\pi|} \prod_{A \in \pi} \kappa[y_A] \, .
\label{wick_cum}
\end{align}
\item If $I':=\cancelm{1}{I}$ denotes the sequence obtained by cancelling the 
first element of $I$, 
\begin{equation}\label{recursion_prop}
\w{y^I}= y_{i_1} \w{y^{I'}}\ - \sum_{(1,i_1)\in V \subset I} \kappa[y_V] \w{y^{I 
\backslash V}}
= y_{i_1} \w{y^{I'}}\ - \sum_{U \subset I'} \kappa[y_{(i_1)+U}] \w{y^{I' 
\backslash U}}\, .
\end{equation}
\end{enumerate}
\end{proposition}

\begin{proof}
{\it Item 1:} The first equality in (\ref{inv_wick}) follows directly from the 
definition (\ref{eq:defTJ}) since then
\begin{align}
{y^I} = \sum_{E\subset I} \E[y^E] \w{y^{I\setminus E}} =  \sum_{U\subset I} 
\E[y^{I\setminus U}] \w{y^U} \, .
\end{align}
The second equality follows then by using the moments-to-cumulants expansion.

{\it Item 2:}
The permutation invariance of the Wick polynomials follows using straightforward 
induction in the definition 
(\ref{eq:defTJ}) since the random variables commute and hence the powers $y^E$ 
are always permutation invariant.

{\it Items 3 and 4:}
Let us first define $ \widetilde{\pazocal{W}}[y^I]$ by setting it equal to the 
right hand side of (\ref{wick_cum})
for any $I\subset I_0$.  If $I=\emptyset$, we have $\widetilde{\pazocal{W}}[y^I] 
= 1 = \w{y^I}$.   Suppose $I \neq \emptyset$.
Since $\sum_{\pi \in \pazocal{P}(\emptyset)}(-1)^{|\pi|} \prod_{A \in \pi} \kappa[y_A]  =1$, 
the definition yields a polynomial of order $|I|$ in $y$.  Our goal is to prove 
that
\begin{align}
\label{check:deriv_wick}
 \partial_{y_j}\widetilde{\pazocal{W}}[y^I] & = \sum_{k=1}^{|I|} \cf(i_k=j)  
\widetilde{\pazocal{W}}\big[y^{\widehat{I}^{(k)}}\big]\, ,\\
\label{recursion2}
 \widetilde{\pazocal{W}}\big[y^I \big] 
& = y_{i_1} \widetilde{\pazocal{W}}\big[y^{\cancelm{1}{I}} \big] - \kappa[y_I] 
 - \sum_{(1,i_1)\in V \subsetneq I}  \kappa[y_V] 
\widetilde{\pazocal{W}}\big[y^{I \backslash V} \big] \,.
\end{align}
Then the claim $\widetilde{\pazocal{W}}[y^I]=\w{y^{I}}$ follows by 
straightforward induction in $|I|$:
Case $|I|=0$ was proven above.  Suppose $I \neq \emptyset$ and that 
$\widetilde{\pazocal{W}}[y^{I'}]=\w{y^{I'}}$ whenever $|I'|<|I|$.
Then the induction assumption and Theorem \ref{th:W_prop} can be used to 
evaluate the expectation of the right hand side of (\ref{recursion2}), 
implying $\E\big[\widetilde{\pazocal{W}}[y^I]\big]=0$.  By 
(\ref{check:deriv_wick}), the polynomial $\widetilde{\pazocal{W}}[y^I]$ also 
satisfies the third defining condition of the Wick polynomials, equation 
(\ref{eq:Appell2}).  Hence, $\widetilde{\pazocal{W}}[y^I]=\w{y^{I}}$ which completes 
the induction
step.  Then (\ref{recursion2}) implies the first identity in 
(\ref{recursion_prop}) and the second identity is found by a relabeling of the 
summation variable.  
Hence, also item 4 follows.

To prove (\ref{check:deriv_wick}), consider some $I \neq \emptyset$.  In the definition of 
$\widetilde{\pazocal{W}}[y^I]$, 
we can express the derivatives of $y^U$, $U\subset I$, as in (\ref{eq:WyIeqs}).  
This shows that 
\begin{align}\label{eq:deriv_wtilde}
& \partial_{y_j}\widetilde{\pazocal{W}}[y^I] = 
\sum_{U \subset I}
 \sum_{k=1}^{|I|} \cf(i_k=j)\, \cf((k,i_k)\in U) y^{\widehat{U}^{(k)}} \sum_{\pi 
\in \pazocal{P}(I \backslash U)}(-1)^{|\pi|} \prod_{A \in \pi} \kappa[y_A] 
\nonumber \\ 
& \quad =  \sum_{k=1}^{|I|} \cf(i_k=j)  \sum_{V \subset \cancelm{k}{I}} y^{V} 
\sum_{\pi \in \pazocal{P}(\widehat{I}^{(k)} \backslash V)}(-1)^{|\pi|} \prod_{A 
\in \pi} \kappa[y_A] \nonumber \\
& \quad = \sum_{k=1}^{|I|} \cf(i_k=j)  
\widetilde{\pazocal{W}}\big[y^{\widehat{I}^{(k)}}\big]\, .
\end{align}
Therefore, (\ref{check:deriv_wick}) holds.

To prove (\ref{recursion2}), denote $x:=(1,i_1)$ and $I':=\cancelm{1}{I}$.  We first split the 
definition into two parts as follows:
\begin{align}\label{eq:splitWatx}
& \widetilde{\pazocal{W}}\big[y^I \big] = \sum_{U \subset I} \cf(x \in U)y^U 
\sum_{\pi \in \pazocal{P}(I \backslash U)}(-1)^{|\pi|} \prod_{A \in \pi} 
\kappa[y_A]
+ \sum_{U \subset I} \cf(x\not\in U)y^U \sum_{\pi \in \pazocal{P}(I \backslash 
U)}(-1)^{|\pi|} \prod_{A \in \pi} \kappa[y_A] \, .
\end{align} 
Following a reasoning similar to (\ref{eq:deriv_wtilde}), we find that the first 
term in the sum on the right hand side 
is equal to $y_{i_1} \widetilde{\pazocal{W}}\big[y^{I'} \big]$.
The second term is equal to 
\begin{align}
 & \sum_{U \subset I} \cf(x \not\in U)y^U 
\sum_{\pi \in \pazocal{P}(I \backslash U)}(-1)^{|\pi|} \bigg(\prod_{x\not \in A 
\in \pi} \kappa[y_A] \bigg) \kappa[y_V] \big\vert_{x \in V \in \pi}
\nonumber \\ & \quad
= - \sum_{U \subset I} \cf(x\not\in U)y^U \sum_{x\in V\subset I\setminus U} 
\kappa[y_V]
\sum_{\pi \in \pazocal{P}((I \backslash U) \backslash V)}(-1)^{|\pi|} \prod_{A 
\in \pi } \kappa[y_A]
\nonumber \\ & \quad
= - \sum_{x\in V\subset I} \kappa[y_V] \sum_{U \subset I \backslash V} y^U
\sum_{\pi \in \pazocal{P}((I \backslash V )\backslash U)}(-1)^{|\pi|} \prod_{A 
\in \pi } \kappa[y_A]
\nonumber \\ & \quad
= - \sum_{x\in V\subset I}  \kappa[y_V] \widetilde{\pazocal{W}}\big[y^{I 
\backslash V} \big] 
 \,.
\end{align} 
Therefore, (\ref{eq:splitWatx}) implies that also (\ref{recursion2}) holds.  This completes the proof of the Theorem.
\end{proof}

\begin{example}
Written is terms of cumulants, the Wick polynomials of lowest order are
\begin{eqnarray}
 \w{y} & =& y - \kappa(y)\, , \nonumber \\
 \w{y_1y_2} &=& y_1 y_2-\kappa(y_1,y_2)-\kappa(y_1)y_2 
-\kappa(y_2)y_1+\kappa(y_1)\kappa(y_2)\, , \nonumber \\
 \w{y_1y_2y_3} &=& y_1 y_2 y_3-\kappa(y_1,y_2,y_3)
 +\kappa(y_1,y_2)\kappa(y_3) +\kappa(y_1,y_3)\kappa(y_2) 
+\kappa(y_2,y_3)\kappa(y_1)
 \nonumber \\ \nonumber
 && - \kappa(y_1)\kappa(y_2)\kappa(y_3)
- \kappa(y_1,y_2) y_3 -\kappa(y_1,y_3) y_2 -\kappa(y_2,y_3) y_1
 \nonumber \\ \nonumber
 &&
 + \kappa(y_1)\kappa(y_2) y_3 +\kappa(y_1)\kappa(y_3) y_2+\kappa(y_2)\kappa(y_3)  y_1
 \nonumber \\ \nonumber
 &&
- \kappa(y_1)y_2 y_3 - \kappa(y_2)y_1 y_3 - \kappa(y_3)y_1 y_2\, .
\end{eqnarray}
\end{example}

\begin{proposition}\label{eq:Wickmultlin}
The Wick polynomials are multilinear, i.e., if $\alpha, \beta$ are constants 
such that $y_j = \alpha y_i + \beta y_{i'}$ for some $j,i,i'\in J$,
then, whenever $I$ and $k$ are such that $i_k=j$, we have
$$
\w{y^{I}} = \alpha \w{y^{\cancelm{k}{I}+(i)}} + \beta 
\w{y^{\cancelm{k}{I}+(i')}} \, .
$$
\end{proposition}
\begin{proof}
The claim follows using multilinearity of cumulants in the representation 
formula (\ref{wick_cum}).
\end{proof}

The following result extends the earlier theorem and shows that multiple 
application of Wick polynomial replacements continues to simplify 
the moments-to-cumulants formula by removing all terms with any internal 
clusters.
\begin{proposition}\label{th:W_mult_prop}
Assume that the measure $\mu$ has all moments of order $N$, i.e.,  suppose that 
$\E[|y^I|]<\infty$ for all $I\in \mathcal{I}$ with 
$|I|\le N$.  Suppose $L\ge 1$ is given and
consider a collection of $L+1$ index sequences
$J',J_\ell\in\mathcal{I}$, $\ell=1,\ldots,L$, such that $|J'|+\sum_\ell 
|J_\ell|\le N$.
Then for $I:= \sum_{\ell=1}^L J_\ell + J'$ 
(with the implicit identification of $J_\ell$ and $J'$ with the set of its 
labels in $I$) we have 
\begin{equation}
\E \bigg[ \prod_{\ell=1}^L \w{y^{J_\ell} }  y^{J'}\bigg] = \sum_{\pi \in 
\pazocal{P}(I)}\prod_{A \in \pi}
\!\left(\kappa[y_A] \cf(A \not\subset J_\ell\  \forall \ell)\right) \, .
\label{wick_prod_multi}
\end{equation}
\end{proposition}
\begin{proof}
We proceed via a double induction: the first induction is over $L$ and the 
second induction is over $|J_L|$. 
The case $L=1$ follows directly from Theorem \ref{th:W_prop}.  Now we assume as 
induction hypothesis of the first level induction that
$L\ge 2$ and 
\begin{eqnarray}
\E \bigg[ \prod_{\ell=1}^{L-1} \w{y^{J_\ell}}  y^{J'} \bigg] = \sum_{\pi \in 
\pazocal{P}(\sum_{\ell=1}^{L-1} J_\ell+J')}\prod_{A \in \pi}
\!\left(\kappa[y_A] \cf(A \not\subset J_\ell\  \forall \ell<L)\right) \,.
\label{ind_hypoth}
\end{eqnarray}
Then we consider the second induction over $ |J_L| =: m$.  For $m=0$, we have 
$\w{y^{J_L}}=1$ and thus then the induction hypothesis (\ref{ind_hypoth}) 
directly
implies (\ref{wick_prod_multi}).
As induction step of the second level hypothesis we take that, for fixed $L$, 
equation (\ref{wick_prod_multi}) holds for all $|J_L|<m$.
Then, if $|J_L|=m$, we can use (\ref{eq:defTJ}), (\ref{ind_hypoth}) and the 
second level induction hypothesis to justify the following argument analogous 
to the one used in the proof of Theorem \ref{th:W_prop}
\begin{align}
& \E \bigg[ \prod_{\ell=1}^L \w{y^{J_\ell}} y^{J'} \bigg]  = 
\E \bigg[ \prod_{\ell=1}^{L-1} \w{y^{J_\ell}}   y^{J_L+J'}\bigg] - 
\sum_{\emptyset\ne E\subset J_L} \E[y^E]\, \E \bigg[
\prod_{\ell=1}^{L-1} \w{y^{J_\ell}}  \w{y^{J_L \backslash E}} y^{J'} \bigg] 
\nonumber \\ & \quad
= 
\sum_{\pi \in \pazocal{P}(\sum_{\ell=1}^{L-1} J_\ell+J_L+J')}\prod_{A \in \pi}
\!\left(\kappa[y_A] \cf(A \not\subset J_\ell\  \forall \ell<L) \right)
\nonumber \\ & \qquad
 - \sum_{\emptyset\ne E\subset J_L} \sum_{\pi_0 \in\pazocal{P}(E)} \prod_{A\in 
\pi_0} \kappa[y_A]
 \sum_{\pi_1 \in \pazocal{P}(\sum_{\ell=1}^{L-1} J_\ell+(J_L\setminus E)+J')} 
\prod_{A' \in \pi_1} \biggl(\kappa[y_{A'}]
\nonumber \\ & \qquad\quad
 \times \cf(A \not\subset J_\ell\  \forall \ell<L) \cf(A'\not\subset 
J_L\setminus E) \biggr)
\nonumber \\ & \quad
= \sum_{\pi \in \pazocal{P}(I)} \prod_{A \in \pi}\!\left(\kappa[y_A] \cf(A 
\not\subset J_\ell\  \forall \ell)\right)
\nonumber \\ & \qquad
+ \sum_{\pi \in \pazocal{P}(I)}\cf(\exists A\in \pi\text{ s.t. }A \subset J_L)  
\prod_{A \in \pi}
\!\left(\kappa[y_A] \cf(A \not\subset J_\ell\  \forall \ell<L)\right)
\nonumber \\ & \qquad
 - \sum_{\pi \in \pazocal{P}(I)} \cf(\cup \defset{A\in \pi}{A\subset J_L}\ne 
\emptyset)
 \prod_{A \in \pi}\!\left(\kappa[y_A] \cf(A \not\subset J_\ell\  \forall 
\ell<L)\right)
\nonumber \\ & \quad
= \sum_{\pi \in \pazocal{P}(I)} \prod_{A \in \pi}\!\left(\kappa[y_A] \cf(A 
\not\subset J_\ell\  \forall \ell)\right) \, .
\end{align}
This completes the induction step and hence also the proof.
 \end{proof}

\section{Cumulants and Wick polynomials as dynamical variables}\label{sec:cwdyn}

To study the state of a random system, it is often better to use the cumulants
rather than the moments of the random
variables.  For instance, if $y,z$ are independent
random variables we have $\E[y^n z^m]= \E[y^n] \E[z^m]$, which is typically 
nonzero, whereas the corresponding cumulant
is zero whenever both $n,m\ne 0$.
Hence, for systems where two regions become ``asymptotically independent'' (for 
instance, for sufficiently mixing
stochastic processes), it is the cumulants, not moments,
which will vanish in the limit. 

To have a concrete example, let us consider a random lattice field $y_x$, 
$x\in\Z^d$, which is spatially sufficiently
strongly mixing.  
Then, for instance, $\kappa(y_0,y_I)\to 0$ if the distance of the index set 
$I\subset \Z^d$ from the origin becomes
unbounded.  Often in the applications the mixing is so strong that
the cumulants,
in this case also called connected correlation functions, become 
``$\ell_1$-clustering'': for any order $n$ one
requires that
$\sup_{x\in \Z^d}\sum_{I\in (\Z^d)^{n-1}} |\kappa(y_x,y_I)|<\infty$.
Naturally, such a property should then not be expected to hold for standard 
correlation functions $\E[y^I]$, apart from
some degenerate examples.

In addition to being mixing, the random fields found in the applications are 
often \defem{translation invariant}.  This means in particular that all 
moments $\E[y^I]$ remain invariant if every index in $I$ is translated by a 
fixed amount, i.e., $\E[y^{I(x)}]=\E[y^I]$ for every $x\in \Z^d$ if we set
$I(x)_\ell := i_\ell-x$.
If the system is both $\ell_1$-clustering and translation invariant,
the cumulants of the Fourier transformed field $\FT{y}_k:= \sum_{x\in \Z^d} 
\rme^{-\ci 2\pi k\cdot x} y_x$, $k$ indexed by the $d$-torus $\T^d$, satisfy
\begin{align}\label{trans_inv_cum}
 \kappa[\FT{y}_{(k_1,k_2,\ldots,k_n)}] = \delta\Bigl(\sum_{\ell=1}^n 
k_\ell\Bigr) \FT{F}_n(k_1,k_2,\ldots,k_n) \, ,
\end{align}
where ``$\delta$'' denotes the Dirac delta distribution and the arithmetic on 
$\T^d$ is defined via periodic identification. Here $\FT{F}_n$
denotes the Fourier transform of $F_n(X):=\cf(X_1=0)\kappa(y_{X})$, $X\in 
(\Z^d)^n$, and for $\ell_1$-clustering measures $\FT{F}_n$
is a uniformly bounded continuous function of $k$.
Therefore, although the cumulants are singular, their singularity structure is 
simple, entirely encoded in the $\delta$-multiplier.
In contrast, by the moments-to-cumulants formula, then for $I:=(1,2,\ldots,n)$ 
and any $k\in (\T^d)^n$
\begin{align}
\mathbb{E}[\FT{y}^{k_I}] = \sum_{\pi \in\pazocal{P}(I)} \prod_{A\in \pi} 
\bigg(\delta\Bigl(\sum_{j\in A} k_j\Bigr)
\FT{F}_{|A|}(k_A)\biggr) \,,
\end{align}
which has ever more complicated singularity structure as the order of the moment 
is increased.
(The above discussion can be made mathematically rigorous by replacing the 
infinite lattice by a periodic $d$-dimensional lattice.
See \cite{LuSp11} for details.)

Hence, for stochastic processes which lead towards a state which is mixing and 
translation invariant,
it seems better to focus on the time-evolution of cumulants instead of the 
corresponding moments.
However, it is not immediately clear how to avoid resorting to the moments as a 
middle step.
It turns out that using Wick polynomials instead of ``bare'' monomials to define 
the time-evolution helps in achieving this goal.
Recall that any monomial term $y^I$
can always be expanded in terms of Wick polynomials using (\ref{inv_wick}), albeit with 
state-dependent coefficients.

To have a concrete example how this could work in practice, we consider in the 
following the case of deterministic evolution with random initial data.
Explicitly, 
we assume that the system is described by random variables $y_j(t)$, where $j$ 
belongs to a fixed (finite) index set $J$, and $t\ge 0$ denotes time.
The initial values $y_j(0)$ are random with a joint distribution $\mu_0$, 
and for each realization of $y(0)$ we assume that the 
values at later times $t>0$ are determined from the
solution to the differential equation
\begin{align}\label{eq:basicev}
 \partial_{t}y_j(t) = \sum_{I\in \mathcal{I}_j} M^I_j(t) \w{y(t)^I}\, ,
\end{align}
where the functions $M^I_j(t)$ are ``interaction amplitudes'' from the $I$:th Wick 
polynomial of $y(t)$ to $y_j(t)$. For each  $j\in J$, the set
$\mathcal{I}_j$ collects those $I\in \mathcal{I}$ which have a nonzero 
amplitude, i.e.,
$M^I_j(t)\ne 0$ for some $t>0$.  
For simplicity, we assume here that $\mathcal{I}_j$
is finite and that the amplitudes $M_j^I(t)$ are some fixed functions of time.  
(They typically might depend on the cumulants of $y(t)$, but this is not
relevant for the discussion below: it suffices that they are not random 
variables.)

We present a concrete example of such a dynamical system in Appendix 
\ref{sec:App_Ham} where we show how the evolution of 
$N$ classical particles interacting 
via a polynomial interaction potential 
can be described by a system of this type assuming a known random distribution of 
initial positions and momenta.  Another explicit example is given in Section 
\ref{sec:DNLS}.

The usefulness of representing the dynamics in the form (\ref{eq:basicev}) 
becomes apparent when looking at the evolution of cumulants.
To avoid technical complications, let us suppose that the joint exponential 
moments of 
$y(t)$ exist and are continuously differentiable and uniformly bounded functions 
of $t$.  This will simplify
the discussion since it allows using the generating functions 
defined earlier in (\ref{eq:defallgenf}).  
With the shorthand notation $\lambda \cdot y := \sum_{j\in J} \lambda_j y_j$, 
the cumulant generating function of $y(t)$, $t$ fixed, is
$g_t(\lambda) = \ln \mean{\rme^{\lambda \cdot y(t)}}$, 
and the Wick polynomial generating function is $G(\lambda;y(t))= 
\frac{\rme^{\lambda\cdot y(t)}}{\E[\rme^{\lambda\cdot y(t)}]} =
\rme^{\lambda\cdot y(t)-g_t(\lambda)}$.  The time-evolution of the cumulant 
generating function is connected
to the Wick polynomial generating function by
\begin{eqnarray}
 \partial_t g_t(\lambda) = \mean{\rme^{\lambda \cdot y(t)}}^{-1} 
\mean{\lambda\cdot \partial_t y(t)\, \rme^{\lambda \cdot y(t)} }
 = \mean{ \lambda\cdot \partial_t y(t)\, G(\lambda;y(t))}\, .
\end{eqnarray}
Therefore, for any $I'\ne \emptyset$ (using the slightly symbolic notations 
``$I\setminus i$'' and ``$y_i$'' instead of
$\cancelm{k}{I}$ and $y_{i_k}$ when  $i=(k,i_k)\in I$)
\begin{align}
 \partial_t \kappa[y(t)_{I'}] = \partial^{I'}_{\lambda}  \partial_t g_t(\lambda) 
\big\vert_{\mathbf{\lambda}=0}
 = \sum_{i \in I'} \mean{ \partial_t y_i(t) \w{y(t)^{I'\setminus i}} } \, .
\end{align}
Hence, if the evolution satisfies (\ref{eq:basicev}), we obtain
\begin{eqnarray}
 \partial_t \kappa[y(t)_{I'}] = \sum_{i \in I'} \sum_{I\in \mathcal{I}_i} 
M^I_i(t)
 \mean{ \w{y(t)^I} \w{y(t)^{I'\setminus i}} } \, .
 \label{cum_evolution}
\end{eqnarray}
In this case, determining the evolution of expectation values of all 
multiplications
of two Wick products, $\mean{ \w{y(t)^{I_1}} \w{y(t)^{I_2}} }$, where both 
$I_1,I_2$ are non-empty, would also yield a solution to the evolution of 
cumulants.

We can now obtain a closed cumulant evolution hierarchy using 
(\ref{cum_evolution}) and Theorems \ref{th:W_prop} and \ref{th:W_mult_prop}.
First, note that for any $I\ne \emptyset$ and $j\in J$ we have
\begin{eqnarray}\label{eq:ktowprod}
 \mean{  \w{y(t)^{I}} \w{y_j(t)} }
 =
 \mean{ \w{y(t)^{I}} y_j(t)  }
 =\kappa[y(t)_{I+(j)}]
   \, ,
 \label{cum_wick}
\end{eqnarray}
since in this case there is exactly one non-internal cluster, the entire set 
$I+(j)$. 
In addition, if $I=\emptyset$, we clearly have $\mean{ \w{y(t)^I} \w{y(t)^{I'}} 
}=\cf(I'=\emptyset)$.  Therefore, the first
two cumulants satisfy, for arbitrary $j,j'\in J$,
\begin{align} \label{eq:cum_hierarchy1}
  \partial_t \kappa[y_j(t)] & =  \cf(\emptyset \in \mathcal{I}_j) 
M^\emptyset_j(t)\, , \\
  \partial_t \kappa[y(t)_{(j,j')}]&  =
  \sum_{\emptyset \ne I\in \mathcal{I}_{j}} 
M^I_{j}(t)\kappa[y(t)_{I+(j')}]+\sum_{\emptyset \ne I\in \mathcal{I}_{j'}} 
M^I_{j'}(t)\kappa[y(t)_{I+(j)}]\, .
\end{align}
For higher order cumulants, with $|I'|\ge 3$,
the equation typically becomes nonlinear; we then have
\begin{align} \label{eq:cum_hierarchy2}
  \partial_t \kappa[y(t)_{I'}]&  =
  \sum_{i \in I'} \sum_{j\in J} \cf((j)\in \mathcal{I}_i) M^{(j)}_i(t) 
\kappa[y(t)_{(j)+(I'\setminus i)}]
  +  \sum_{i \in I'} 
  \sum_{I\in \mathcal{I}_i, |I|\ge 2} M^I_i(t)
  \mean{ \w{y(t)^I} \w{y(t)^{I'\setminus i}} } \, .
\end{align}
We have separated here the terms with $|I|=1$ to show how they operate linearly on 
the cumulants of order $|I'|$
(note that $\kappa[y(t)_{(j)+(I'\setminus i)}]=\kappa[y(t)_{I''}]$ for the 
sequence $I''$ which is obtained from $I'$
by replacing $i$ with $j$).
In the final sum, both $|I|$ and $|I'\setminus i|$ are greater than one, so it has a cumulant
expansion 
\begin{align}
 \mean{ \w{y(t)^I} \w{y(t)^{I'\setminus i}} } 
 = \sum_{\pi \in \pazocal{P}(I+(I'\setminus i))} \prod_{A \in \pi}
\!\left(\kappa[y(t)_A] \cf(A\cap I \ne \emptyset, \
A\cap (I'\setminus i) \ne \emptyset)\right) \, ,
\end{align}
i.e., all clusters have to contain at least one element from both sequences.  In 
particular, 
it cannot contain any singlets, i.e., it does not depend on any of 
$\kappa[y_j(t)]$, $j\in J$.  
Let us also point out that since in these terms $|I|+|I'\setminus i|>|I'|$, any 
linear term is 
necessarily of higher order.  In particular, this means that lower order 
cumulants can appear only in nonlinear combinations.

Instead of studying the full cumulant hierarchy, one can also use evolution 
estimates for the Wick polynomials.
The situation often encountered in the applications is that the properties of 
the initial measure are fairly well known, whereas
very little a priori control exists for the time-evolved measure.  In such a 
case, one can use the above result and obtain a perturbation 
expansion by applying the fundamental theorem of calculus.  With the shorthand 
notation $y:=y(0)$ we have
\begin{eqnarray}
\kappa[y(t)_{I'}]&=& \kappa[y_{I'}] + \int_0^t ds \sum_{i \in I'} \sum_{I\in 
\mathcal{I}_i} M^I_i(s)
 \mean{ \w{y(s)^I} \w{y(s)^{I'\setminus i}} } \nonumber \\
 &=& \kappa[y_{I'}] + \sum_{i \in I'} \sum_{I\in \mathcal{I}_i}
 \mean{ \w{y^I} \w{y^{I'\setminus i}} } \int_0^t ds M^I_i(s) \nonumber \\
 && + \sum_{i \in I'} \sum_{I\in \mathcal{I}_i}  \int_0^{t}ds'
 \partial_{s'} \mean{ \w{y(s')^I} \w{y(s')^{I'\setminus i}} }
 \int_{s'}^t ds  M^I_i(s) \, ,
 \label{cum_expansion}
\end{eqnarray}
where we have applied Fubini's theorem to the final integral.  This type of 
expansion could be helpful if
the coefficients $\int_0^t ds M^I_i(s)$ behave better than $M^I_i(t)$, such as 
in the presence of fast oscillations.
Further iterations of this procedure,
using either the above cumulant hierarchy or any of the Wick polynomial 
hierarchies below,
would then yield an expansion of $\kappa[y(t)_{I'}]$ in terms of the
expectations at time $t=0$ and the time dependent amplitudes $ M^I(t)$.  This is 
particularly useful if all $M^I(t)$ are small, since each iteration adds
one more such factor.

Let us conclude this section by deriving recursion formulae for the products of 
Wick polynomials.  As mentioned earlier, these could then
be used instead of the direct cumulant hierarchy to study the time-evolution of 
the cumulants.
For this, it would suffice to study $ \mean{ \w{y(s')^I} \w{y(s')^{I'\setminus 
i}} }$ appearing in (\ref{cum_evolution}), but typically
the products of two terms do not satisfy a closed evolution equation and a full 
hierarchy will be needed.
Let us begin with the evolution equation for $ \w{y(t)^I}$.  For any 
deterministic evolution
process, we can obtain a fairly compact evolution equation by treating the 
time-derivative $\partial_t{y}_j$ as a new
random variable:
\begin{eqnarray}
\partial_t \w{y(t)^I} = \partial^{I}_{\lambda}\partial_t G(\lambda;y(t)) 
\big\vert_{\mathbf{\lambda}=0}
= \sum_{i \in I} \w{\,(\partial_t y_i(t)) y(t)^{I \backslash i}} \, .
\label{wick_evolution1}
\end{eqnarray}
The form is analogous to the standard Leibniz rule.  For products of Wick 
polynomials, we thus have
\begin{eqnarray}
\partial_t \prod_{k=1}^n\w{y(t)^{I_k}}
= \sum_{k=1}^n \sum_{i \in I_k} \w{\,(\partial_t y_i(t)) y(t)^{I_k \backslash 
i}} \prod_{k'\ne k} \w{y(t)^{I_{k'}}} \, .
\end{eqnarray}
Assuming (\ref{eq:basicev}) and using multilinearity, we then obtain the 
following equation involving
``nested Wick products'':
\begin{eqnarray}\label{eq:nestedWick}
\partial_t \E \bigg[ \prod_{k=1}^n\w{y(t)^{I_k}} \bigg]
= \sum_{k=1}^n \sum_{i \in I_k} \sum_{I \in \mathcal{I}_i} M_i^I(t)\E \bigg[ 
\w{\,(\w{y(t)^I})
 y(t)^{I_k \backslash i}} \prod_{k'\ne k} \w{y(t)^{I_{k'}}} \, \bigg] .
\label{nested_wick}
\end{eqnarray}

The formula (\ref{eq:nestedWick})
is appealing in its simplicity but it does not directly lead to closed hierarchy 
of equations.  This can be achieved by expanding the
nested product in terms of cumulants and Wick products.  For this, let us note 
that by (\ref{recursion_prop}) and the
observation made after (\ref{eq:ktowprod}),
we have for any $I'$
\begin{eqnarray}
\w{\,(\partial_t y_i(t)) y(t)^{I'}} =
\partial_t y_i(t) \w{y(t)^{I'}}\ -\ \sum_{U \subset I'} \E[\partial_t y_i(t) 
\w{y(t)^U}] \w{y(t)^{I' \backslash U}} \, .
\end{eqnarray}
Therefore, whenever (\ref{eq:basicev}) holds, we find that
\begin{align}\label{second_hierarchy}
 & \partial_t  \E\Bigl[\prod_{k=1}^n\w{y(t)^{I_k}}\Bigr] 
   = \sum_{k=1}^n \sum_{i \in I_k} \sum_{I_{n+1}\in \mathcal{I}_i} 
M^{I_{n+1}}_i(t) \biggl\lbrace
   \E\Bigl[\w{y(t)^{I_k \backslash i}} \prod_{k'=1;k'\ne 
k}^{n+1}\w{y(t)^{I_{k'}}}\Bigr]
   \nonumber \\ & \qquad
   -  \sum_{ U \subset I_k \backslash i}   \E\Bigl[\w{y(t)^{(I_k \backslash i) 
\backslash U}} \prod_{k'=1;k'\ne k}^{n}\w{y(t)^{I_{k'}}}\Bigr]
   \ \E\Bigl[\w{y(t)^{I_{n+1}}}\w{y(t)^{U}}\Bigr]
\biggr\rbrace
\end{align}
This forms a closed hierarchy of evolution equations for the collection of all
expectation values of the type $\E\Bigl[\prod_{k=1}^n\w{y(t)^{I_k}}\Bigr]$.

A second alternative for the hierarchy follows from the observation that if 
$y(t)$ and $z(t)$ are two processes which start
with independent, identically distributed initial data, then at any later moment 
they are also independent and identically distributed and hence
\begin{align}
& \partial_t G(\lambda;y(t)) = G(\lambda;y(t)) \left(\lambda\cdot \partial_t 
y(t) - \E_z[\lambda\cdot \partial_t z(t) \,G(\lambda;z(t)) ] \right)
\\ \nonumber & \quad 
=\E_z\left[G(\lambda;y(t))G(\lambda;z(t)) \lambda\cdot(\partial_t y(t) 
-\partial_t z(t) ) \right]
\end{align}
where in the second equality we have used $\E_z[G(\lambda;z(t))]=1$.
Consider then the product measure for the processes $y,z$ and
let $G'$ denote the corresponding Wick polynomial generating function.
Since by Fubini's theorem then $\E_{y,z}[\rme^{\lambda\cdot(y(t)+z(t))}]= 
\E_{y}[\rme^{\lambda\cdot y(t)}]^2$, now
$G(\lambda;y(t))G(\lambda;z(t)) = G'(\lambda;y(t)+z(t))$.
Hence, for dynamics satisfying (\ref{eq:basicev}) 
\begin{align}
 & \partial_t  \E_y\!\left[\prod_{k=1}^n\w{y(t)^{I_k}}\right] 
   = \sum_{k=1}^n \sum_{i \in I_k} \sum_{I\in \mathcal{I}_i} M^I_i(t) 
   \E_{y,z}\Bigl[(\w{y(t)^I} -\w{z(t)^I} )\ \w{\, (y(t)+z(t))^{I_k\setminus i}} 
\prod_{k'\ne k} \w{y(t)^{I_{k'}}} \Bigr]\, .
   \label{third_hierarchy_1}
\end{align}
Let us point out that the earlier expression in (\ref{second_hierarchy}) follows 
from the above one if we expand the power $(y(t)+z(t))^{I_k\setminus i}$
and then use the fact that the joint measure is a product measure.
The formula does not yet yield a closed hierarchy but the following 
generalization does so: if $z_{k,\ell}(t)$ are processes such that
their joint initial distribution is given, then 
\begin{align}
 & \partial_t  \E\!\left[\prod_{k=1}^n\w{\,\Bigl(\sum_{\ell} 
z_{k,\ell}(t)\Bigr)^{\!I_k}}\right]
   = \sum_{k=1}^n \sum_{i \in I_k} \sum_{I\in \mathcal{I}_i} \sum_{\ell} 
M^I_i(t) 
      \nonumber \\ & \qquad
   \times \E'_k\Bigl[(\w{z_{k,\ell}(t)^I} -\w{z'_{k,\ell}(t)^I} )\
   \w{\,\Bigl(\sum_{\ell'} (z_{k,\ell'}(t)+z'_{k,\ell'}(t))\Bigr)^{I_k\setminus 
i}} \prod_{k'\ne k}
   \w{\,\Bigl(\sum_{\ell'} z_{k',\ell'}(t)\Bigr)^{\!I_{k'}}} \Bigr]\, ,
   \label{third_hierarchy_2}
\end{align}
where $\E'_k$ refers to a measure where $z_{k,\ell}$ for each $\ell$ has been 
independently duplicated in the $z'_{k,\ell}$-process.

A possible benefit of this formulation could be when $z_{k,\ell}(t)$ have mean
zero and are independent for all $\ell$.   Then the central limit theorem 
governs the behavior of $\sum_{\ell} z_{k,\ell}$
when there are many terms in the sum.
Therefore, it could be of help in controlling the 
otherwise difficult case where one has performed many iterations 
starting from (\ref{third_hierarchy_1}).

\section{Kinetic theory of the discrete NLS equation revisited}\label{sec:DNLS}

In this section we apply the previous Wick polynomial techniques to the discrete 
nonlinear Schr\"{o}dinger equation.  This example is chosen since it has a particular simple, but 
nontrivial, Wick expansion of the evolution equation.
In addition, we can then rely on the rigorous results and known properties from 
an earlier work on the model \cite{LuSp11}.
We focus on the kinetic theory of the model on the ``kinetic''
time-scale which is $O(\lambda^{-2})$
if the nonlinear coupling $\lambda$ is taken to zero.  

We begin by going through the derivation of the Boltzmann-Peierls equation in 
the spatially homogeneous case,
and show how the task is simplified by using the Wick expanded dynamics and the 
cumulant hierarchy, as explained in Section \ref{sec:cwdyn}.
We only consider terms which would be present in the simplified case of Gaussian initial 
data.  We give an example in Appendix \ref{sec:nonGauss} which highlights the mechanism 
leading to suppression of the additional effects of non-Gaussian initial data 
in the kinetic scaling limit.

For this particular setup, it is easy to find dynamical variables whose 
evolution equation does not have a linear part.  This is an
important simplification since it negates a term which becomes rapidly 
oscillating on the kinetic time-scale, having a divergent
frequency $O(\lambda^{-1})$ in the kinetic scaling limit.  The effect becomes 
apparent when looking at field time-correlations instead of the 
evolution of equal time cumulants.  We discuss the issue in more detail in 
Section \ref{sec:inhomog}.

In \cite{LuSp11}, the initial data is taken to be given by a thermal Gibbs 
measure which is stationary both in space and in time.
We do not assume the initial data to be time-stationary here, but 
the computations in this section require space translation invariance.  The spatially
inhomogeneous case is technically substantially more 
complicated, and we discuss it only briefly in Section \ref{sec:inhomog}.

The results in this section are derived in the spirit of standard perturbation 
theory and focus solely on evolution on short kinetic time
scales, $t=\tau \lambda^{-2}$ with $0<\tau\ll 1$.  
It is however possible to apply the cumulant 
hierarchy differently, leading to equations which do not require 
taking a scaling limit.  We conclude the study of the DNLS model in
Section \ref{sec:beyondkin} by proposing a reformulation of the problem which
leads to Boltzmann type evolution equations which could be accurate also for 
times longer than $O(\lambda^{-2})$.  The discussion is not completely mathematically rigorous, 
but we propose a scheme which could be used to this end under some natural conditions about the time-evolved state.

The discrete NLS equation on the lattice $\Z^d$
deals with functions  $ \psi: \R \times \Z^d \rightarrow \C $ which satisfy
\begin{align}\label{eq:defNLSx}
\ci\partial_t \psi_t(\mathbf{x}) = \sum_{\mathbf{y} \in \Z^d} 
\alpha(\mathbf{x}-\mathbf{y})\psi_t(\mathbf{y}) + \lambda |\psi_t(\mathbf{x})|^2 
\psi_t(\mathbf{x}) \, .
 \end{align}
Here the function $ \alpha: \Z^d \rightarrow \R$ is called the hopping amplitude 
and we assume that it is
symmetric, $\alpha(\mathbf{x})=\alpha(-\mathbf{x})$, and exponentially 
decreasing.  The parameter 
$ \lambda>0$ is considered to be small, and in the kinetic scaling limit we take 
$\lambda\to 0$ 
and $ t = \tau \lambda^{-2}$ with $\tau>0$ fixed.
The initial field $\psi_0$ is assumed to be random, bounded on finite subsets of 
the lattice, 
and to have an $\ell_1$-clustering distribution.  We aim at controlling the 
moments of the random variables $\psi_t(\mathbf{x})$ and $\psi_t(\mathbf{x})^*$ 
which we label using $\psi_t(\mathbf{x},+1):=\psi_t(\mathbf{x})$
and $\psi_t(\mathbf{x},-1):=\psi_t(\mathbf{x})^*$.

Since we do not assume that $\psi_0$ is $\ell_2$-summable, even the (almost 
sure) 
existence of solutions to (\ref{eq:defNLSx}) becomes an issue.  
To our knowledge, it has not been proven for the above setup, and most likely, 
some additional 
assumptions about the increase of the values of the initial field at infinity 
are needed for proper existence theory.
However, these problems can be easily avoided by replacing the infinite lattice 
$\Z^d$ by a finite lattice with 
periodic boundary conditions: see \cite{LuSp11} for details.  This would merely 
result in replacing the lattice, the Fourier space and transform,
and the associated $\delta$-functions by their finite lattice counterparts.  
Since  even then the final limits cannot be rigorously controlled,
we opt here for some additional formality in the discussion, but with less 
complicated formulae to deal with.

For technical simplicity, here we also only consider initial data which are 
``gauge invariant'': we will always suppose
$ \psi_0(\mathbf{x})$ has the same distribution as $ 
\rme^{\ci\theta}\psi_0(\mathbf{x})$ for any $ \theta \in [0,2 \pi]$.  In fact, 
this transformation
commutes with the time evolution, i.e., if the initial field is changed from 
$\psi_0$ to  $ \rme^{\ci\theta}\psi_0$,
the time-evolved field will change from $\psi_t$ to $ \rme^{\ci\theta}\psi_t$.  
In particular, 
also the field $\psi_t$ will then be gauge invariant.  The main reason for 
insisting on this assumption is that it will 
automatically force many cumulants to be zero and hence simplify the 
combinatorics associated with the cumulant hierarchy.
Gauge invariance implies that a moment is zero unless
it has the same number of $\psi$ and $\psi^*$ factors, even when 
the fields are evaluated at different times.  Hence, it implies that 
every odd moment of the fields is zero and hence also every odd cumulant.
Similarly, we see that even cumulants are also zero if they concern a
different number of $ \psi^*$ and $ \psi$ variables.  

For instance, we find using (\ref{wick_cum}) that for any gauge invariant state 
and any $a_j := \psi_{t_j}(\mathbf{x}_j,\sigma_j)$, $j=1,2,3$, 
\begin{align}
\w{a_1 a_2 a_3} = a_1 a_2 a_3 - \E[a_1 a_2] a_3 - \E[a_1 a_3] a_2- \E[a_2 a_3] 
a_1\, .
\end{align}
This is the definition of the ``pairing truncation operation'' 
$\widehat{\pazocal{P}}$ given in Lemma 3.2 of \cite{LuSp11}.
Applying the truncation operation in the evolution equation
was one of the key changes to the standard perturbation theory which allowed the 
rigorous analysis in \cite{LuSp11}.
With the benefit of hindsight, we can now identify it as a Wick contraction of 
the random variables.

Under the above assumptions and using (\ref{inv_wick}), we 
find that (\ref{eq:defNLSx}) is equivalent to the following Wick contracted 
evolution equation
\begin{align}
& \ci\partial_t \psi_t(\mathbf{x}) = \sum_{\mathbf{y} \in \Z^d} 
\alpha(\mathbf{x}-\mathbf{y}) \w{\psi_t(\mathbf{y})} +
 2 \lambda \E[\psi_t(\mathbf{x})^*  \psi_t(\mathbf{x})] \w{\psi_t(\mathbf{x})}
+ \lambda \w{\psi_t(\mathbf{x})^*  \psi_t(\mathbf{x}) \psi_t(\mathbf{x})} \, .
\end{align}
Hence, the random variables $\psi_t(\mathbf{x},\sigma)$, $\sigma=\pm 1$, satisfy 
an evolution equation of a form required in the previous section,
in (\ref{eq:basicev}),
\begin{align}\label{eq:defNLSwick}
& \ci\sigma \partial_t \psi_t(\mathbf{x},\sigma) = \sum_{\mathbf{y} \in \Z^d} 
\alpha(\mathbf{x}-\mathbf{y}) \w{\psi_t(\mathbf{y},\sigma)} +
  \lambda R_t(\mathbf{x}) \w{\psi_t(\mathbf{x},\sigma)}
+ \lambda \w{\psi_t(\mathbf{x},-1)  \psi_t(\mathbf{x},\sigma) 
\psi_t(\mathbf{x},1)} \, ,
\end{align}
where we have defined $R_t(x):= 2 \E[|\psi_t(\mathbf{x})|^2]\ge 0$, which is 
also equal to $2\kappa(\psi_t(\mathbf{x},-1), \psi_t(\mathbf{x},1))$.

\subsection{Translation invariant initial measures}\label{sec:homog}

The evolution problem simplifies significantly, if we assume that 
the initial data is not only gauge, 
but also translation invariant.  Since also 
spatial translations commute with the time evolution, we can use the earlier 
results for the cumulants of Fourier transforms of the random field.  In 
particular, then for any $t\ge 0$ and $\mathbf{x},\mathbf{y}\in \Z^d$,
\begin{align}
\E[\psi_t(\mathbf{x})^*\psi_t(\mathbf{y})]=\E[\psi_t(\mathbf{x}-\mathbf{y}
)^*\psi_t(0)]\,. 
\end{align}
This implies that $R_t(\mathbf{x}) = R_t(0)=: R_t$ for all $t$, and therefore 
the evolution equation (\ref{eq:defNLSwick}) 
for translation and gauge invariant initial data can be written as
\begin{align}\label{eq:NLSti}
& \ci\sigma \partial_t \psi_t(\mathbf{x},\sigma) = \sum_{\mathbf{y} \in \Z^d} 
\alpha^\lambda_t(\mathbf{x}-\mathbf{y}) \w{\psi_t(\mathbf{y},\sigma)} 
+ \lambda \w{\psi_t(\mathbf{x},-1)  \psi_t(\mathbf{x},\sigma) 
\psi_t(\mathbf{x},1)} \, ,
\end{align}
where $\alpha^\lambda_t(\mathbf{x}) := \alpha(\mathbf{x}) + \lambda 
\cf(\mathbf{x}=\mathbf{0}) R_t $.  

Using multilinearity of the Wick polynomials, we thus find the following 
evolution equation for the Fourier 
transformed fields $\FT{\psi}_t(\mathbf{k},\sigma) := \sum_{\mathbf{x}} 
\rme^{-\ci 2 \pi \mathbf{k}\cdot \mathbf{x}} \psi_t(\mathbf{x},\sigma)$,
\begin{align}\label{eq:NLStiFT}
&\partial_t \FT{\psi}_t(\mathbf{k},\sigma)= -\ci \sigma  
\omega_t^\lambda(\mathbf{k})  \w{\FT{\psi}_t(\mathbf{k},\sigma)}
      \nonumber \\ & \qquad
- \ci  \sigma\lambda \int_{(\T^{d})^3} d\mathbf{k}_1 d\mathbf{k}_2 
d\mathbf{k}_3 
 \delta(\mathbf{k}- \mathbf{k}_1 - \mathbf{k}_2 -\mathbf{k}_3) 
\w{\FT{\psi}_t(\mathbf{k}_1,-1)\FT{\psi}_t(\mathbf{k}_2,\sigma)\FT{\psi}
_t(\mathbf{k}_3,1)}\, ,
\end{align}
where 
\begin{align}
 \omega^\lambda_t(\mathbf{k}) := \FT{\alpha}^\lambda_t(\mathbf{k}) = 
\FT{\alpha}(\mathbf{k})+ \lambda R_t \, .
\end{align}
For later use, let us point out that the definitions of the random fields imply 
the following 
rule for complex conjugation of the Fourier transformed fields: 
$\FT{\psi}_t(\mathbf{k},\sigma)^* = \FT{\psi}_t(-\mathbf{k},-\sigma)$.
In addition, the assumed symmetry of $\alpha$ implies the symmetry 
$\omega^\lambda_t(-\mathbf{k})=\omega^\lambda_t(\mathbf{k})$.

We recall that in the present translation invariant setting, the $n$:th 
cumulants satisfy for $k\in (\T^d)^n$, $\sigma\in \set{-1,1}^n$
\begin{align}
 \kappa(\FT{\psi}_t(\mathbf{k}_1,\sigma_1),\ldots, 
\FT{\psi}_t(\mathbf{k}_n,\sigma_n))
 = \delta\Bigl(\sum_{\ell=1}^n \mathbf{k}_\ell\Bigr) \FT{F}_n(k,\sigma;t) \, ,
\end{align}
where $F_n(x,\sigma;t) := \cf(\mathbf{x}_1=0) 
\kappa(\psi_t(\mathbf{x}_1,\sigma_1),\ldots,\psi_t(\mathbf{x}_n,\sigma_n))$
is identically zero unless $\sum_\ell \sigma_\ell=0$.  We are now mainly 
interested in the evolution of the lowest nonzero cumulants, i.e., 
of $F_2(\mathbf{x},(-1,1))$.  We denote its Fourier transform by $W$; more 
precisely, we set
\begin{equation}
W^{\lambda}_t(\mathbf{k}) := \sum_{ \mathbf{x}\in\Z^d}\rme^{-\ci 2 \pi 
\mathbf{k} \cdot \mathbf{x}}\kappa(\psi_t(0)^*, \psi_t( \mathbf{x}))
= \sum_{ \mathbf{x}\in\Z^d}\rme^{-\ci 2 \pi \mathbf{k} \cdot 
\mathbf{x}}\E[\psi_t(0)^* \psi_t( \mathbf{x})] \, .
\label{W_definition}
\end{equation}
It follows that $\FT{F}_2((\mathbf{k}_1,\mathbf{k}_2),(-1,1);t)=  
W^{\lambda}_t(\mathbf{k}_2)$ and 
$\FT{F}_2((\mathbf{k}_1,\mathbf{k}_2),(1,-1);t)=  W^{\lambda}_t(-\mathbf{k}_2)$. 
 Therefore, we have the following general rule for second order cumulants:
\begin{align}\label{eq:kappa2psi}
 \kappa(\FT{\psi}_t(\mathbf{k}_1,\sigma_1), \FT{\psi}_t(\mathbf{k}_2,\sigma_2))
 = \delta(\mathbf{k}_1+\mathbf{k}_2) \cf(\sigma_1+\sigma_2=0) 
W^{\lambda}_t(\sigma_2\mathbf{k}_2)\, .
\end{align}

Therefore, to study the evolution of all second moments in this systems, it 
suffices to study the function $W^{\lambda}_t$. 
In particular, clearly $R_t =2 \int_{\T^{d}}\! 
d\mathbf{k}\,W^{\lambda}_t(\mathbf{k})$.
Note that by translation invariance, $W^{\lambda}_t$ is always real valued.

\subsubsection{Heuristic derivation of the Boltzmann-Peierls equation}\label{derivation_boltzmann}

After these preliminaries, we are ready for 
an application of the cumulant hierarchy to study the Fourier transformed fields.
Our first goal is to justify the Boltzmann-Peierls equation which has been 
conjectured before, based on perturbation expansions in \cite{LuSp11}.
The conjecture says that in the kinetic scaling limit the function $W$ should 
converge: it is assumed that there exists a limit 
\begin{align}
 W_\tau(\mathbf{k}) := \lim_{\lambda\to 0}W^{\lambda}_{\tau\lambda^{-2}}(\mathbf{k})\, .
\end{align}
In addition, the analysis of the perturbation series suggests that the limit 
satisfies the following homogeneous Boltzmann-Peierls equation:
\begin{align}\label{Boltzmann}
\partial_t W_t(\mathbf{k})= \pazocal{C}(W_t(\cdot))(\mathbf{k})\, ,
\end{align}
with the collision operator 
\begin{align}\label{eq:defCBP}
&\pazocal{C}(W(\cdot))(\mathbf{k})= 
4 \pi \int_{(\T^{d})^3} d\mathbf{k}_1d\mathbf{k}_2 d\mathbf{k}_3 
\delta(\mathbf{k}+\mathbf{k}_1-\mathbf{k}_2-\mathbf{k}_3 )\delta(\omega 
+\omega_1 - \omega_2 - \omega_3)
\nonumber \\ & \quad
 \times 
 [W(\mathbf{k}_1)W(\mathbf{k}_2)W(\mathbf{k}_3) 
 + W(\mathbf{k})W(\mathbf{k}_2)W(\mathbf{k}_3) 
 - W(\mathbf{k})W(\mathbf{k}_1)W(\mathbf{k}_3) - 
W(\mathbf{k})W(\mathbf{k}_1)W(\mathbf{k}_2)] 
\end{align}
where $\omega_i := \FT{\alpha}(\mathbf{k}_i)$ and $\omega := \FT{\alpha}(\mathbf{k})$.

In fact, a lucky accident hides the fact that our present random fields are 
actually ill suited for taking of the scaling limit:
it is clear from the linear part in (\ref{eq:NLStiFT}) that they are
highly oscillatory, and only observables where these
oscillations cancel out, can be hoped to have a (nonzero) limiting value 
in the kinetic scaling limit.
Fortunately, there is a simple ``renormalization'' which cancels these fast 
oscillations.  If we define a new random field by the formula
\begin{eqnarray}\label{eq:defatfield}
a_t(\mathbf{k},\sigma) =\FT{\psi}_t(\mathbf{k},\sigma) \exp 
\bigg(\ci\sigma\int_0^t\! ds\, \omega^\lambda_s(\mathbf{k}) \bigg)\, ,
\label{a_def}
\end{eqnarray}
then it clearly satisfies an equation without a linear term.  Explicitly, then 
\begin{align}
& \partial_t a_t(\mathbf{k}_1,\sigma)
=  - \ci\sigma\lambda \int_{(\T^{d})^3} d\mathbf{k}_2 d\mathbf{k}_3 
d\mathbf{k}_4 
\delta(\mathbf{k}_1 - \mathbf{k}_2 - \mathbf{k}_3 -\mathbf{k}_4)
\nonumber \\ & \qquad \times
\rme^{\ci t(\sigma \omega_1 + \omega_2 -\sigma \omega_3 -\omega_4)}
 \w{a_t(\mathbf{k}_2,-1)a_t(\mathbf{k}_3,\sigma)a_t(\mathbf{k}_4,1)} \, . 
\end{align}
Note that due to the alternating signs, the time dependent terms cancel each other out in the
oscillatory phase term inside the integral.
In fact, the same happens in the second order cumulants, as can be checked by 
using (\ref{eq:kappa2psi}) and the symmetry of $\omega^\lambda_s$: we then find 
that 
\begin{align}\label{eq:kappa2a}
 \kappa(a_t(\mathbf{k}_1,\sigma_1), a_t(\mathbf{k}_2,\sigma_2))
 = \delta(\mathbf{k}_1+\mathbf{k}_2) \cf(\sigma_1+\sigma_2=0) 
W^{\lambda}_t(\sigma_2\mathbf{k}_2)\, .
\end{align}
It is clear that multiplication with a nonrandom term as in (\ref{a_def}) does 
not spoil the gauge invariance of the field so we can rely on it also when working with the cumulants
of the $a$-fields.

We can now study the evolution of $W^{\lambda}_t(\mathbf{k})$ by employing
the expansion given in (\ref{cum_expansion}) to the cumulant
$\kappa(a_t(\mathbf{k}',\sigma'), a_t(\mathbf{k},\sigma))$.   We then find using any $\sigma'=-\sigma$ that
\begin{align}\label{W_expansion}
 & \delta(\mathbf{k}'+\mathbf{k})
(W^{\lambda}_t(\sigma \mathbf{k})-W^{\lambda}_0(\sigma \mathbf{k}))
\nonumber \\ &  
=
-\ci\lambda \sigma \int_0^t ds \int_{(\T^{d})^3} d\mathbf{k}_1 
d\mathbf{k}_2 d\mathbf{k}_3 \delta(\mathbf{k} - \mathbf{k}_1 - \mathbf{k}_2 
-\mathbf{k}_3)\rme^{\ci s(\sigma \omega +\omega_1 - \sigma \omega_2 - 
\omega_3)}\kappa[(a_0)_{I}]  
\nonumber \\ & \quad 
-\ci\lambda \sigma' \int_0^t ds \int_{(\T^{d})^3} d\mathbf{k}_1 d\mathbf{k}_2 
d\mathbf{k}_3 \delta(\mathbf{k}' - \mathbf{k}_1 - \mathbf{k}_2 
-\mathbf{k}_3)\rme^{\ci s(\sigma' \omega' +\omega_1 - \sigma' \omega_2 - 
\omega_3)} \kappa[(a_0)_{I'}]  \nonumber \\
& \quad 
-\lambda^2 \sigma \sum_{\ell \in I} \sigma_{\ell} \int_0^t ds \int_{(\T^{d})^3} 
d\mathbf{k}_1 d\mathbf{k}_2 d\mathbf{k}_3 \delta(\mathbf{k} - \mathbf{k}_1 - 
\mathbf{k}_2 -\mathbf{k}_3)\rme^{\ci s(\sigma \omega +\omega_1 - \sigma \omega_2 
- \omega_3)} \nonumber \\
& \qquad \times \int_0^s ds'   \int_{(\T^{d})^3} d\mathbf{k}_4 d\mathbf{k}_5 
d\mathbf{k}_6 \delta(\mathbf{k}_{\ell} - \mathbf{k}_4 - \mathbf{k}_5 
-\mathbf{k}_6)\rme^{\ci s'(\sigma_{\ell} \omega_{\ell} +\omega_4 - \sigma_{\ell} 
\omega_5 - \omega_6)} \E \big[ \w{a_{s'}^{J_{\ell}}} \w{a_{s'}^{\cancelm{\ell}{I}}} \big] 
\nonumber \\ & \quad  
-\lambda^2 \sigma' \sum_{\ell \in I'} \sigma_{\ell} \int_0^t ds 
\int_{(\T^{d})^3} d\mathbf{k}_1 d\mathbf{k}_2 d\mathbf{k}_3 \delta(\mathbf{k}' - 
\mathbf{k}_1 - \mathbf{k}_2 -\mathbf{k}_3)\rme^{\ci s(\sigma' \omega' +\omega_1 - 
\sigma' \omega_2 - \omega_3)} \nonumber \\
& \qquad \times \int_0^s ds'   \int_{(\T^{d})^3} d\mathbf{k}_4 d\mathbf{k}_5 
d\mathbf{k}_6 \delta(\mathbf{k}_{\ell} - \mathbf{k}_4 - \mathbf{k}_5 
-\mathbf{k}_6)\rme^{\ci s'(\sigma_{\ell} \omega_{\ell} +\omega_4 - \sigma_{\ell} 
\omega_5 - \omega_6)} \E \big[ \w{a_{s'}^{J_{\ell}}} \w{a_{s'}^{\cancelm{\ell}{I'}}} \big] 
\,,
\end{align}
where
\begin{align}
\label{eq:defIseq}
& I=\left((\mathbf{k}_1,-1), (\mathbf{k}_2,\sigma), (\mathbf{k}_3,1), 
(\mathbf{k}',\sigma') \right) \,,\\ 
& I'=\left((\mathbf{k}_1,-1), (\mathbf{k}_2,\sigma'), (\mathbf{k}_3,1), 
(\mathbf{k},\sigma) \right) \,, \\
& J_{\ell}=\left((\mathbf{k}_4,-1), (\mathbf{k}_5,\sigma_{\ell}), (\mathbf{k}_6,1) 
\right) \,.
\end{align}

Following the standard perturbation recipe, we next apply the cumulant hierarchy to the terms depending on $a_{s'}$ in (\ref{W_expansion}).
This results in a sum of two terms: one, in which every $a_{s'}$ has been replaced by $a_0$, plus a ``correction'' which we denote by
$\delta(\mathbf{k}' +  \mathbf{k})\pazocal{R}_3(\sigma\mathbf{k},t)$.
Further iterations of the perturbation expansion and a careful study of the oscillations of the term by term expansion as in \cite{LuSp11}
leads us to the conjecture that $\pazocal{R}_3(\mathbf{k},\tau\lambda^{-2})$ 
should converge in the kinetic scaling limit, as $\lambda\to 0$, at least for sufficiently nondegenerate 
dispersion relations and for large enough dimension $d$.  In addition, the analysis indicates that the limit value is $O(\tau^2)$, which is
negligible compared to the contribution from the other terms following from (\ref{W_expansion}).  
However, the term by term analysis does not suffice to actually prove the claim since the method which was used to rigorously 
control the convergence of the perturbation expansion in \cite{LuSp11} was based on time stationarity of the initial state.
This assumption cannot be made here since we are interested in nontrivial time evolution effects.  
Instead of going into the details of the above argument, we discuss a less technically 
involved motivation for the claim in Section \ref{sec:beyondkin}.

Thus now we need to evaluate expectations of the form $ \E[\w{a_1a_2a_3} \w{a_4a_5a_6}]$ 
where each $a_i$ stands for one of the field variables.
The cumulant expansion in Theorem \ref{th:W_mult_prop} and the vanishing of the third order cumulants imply
\begin{align}\label{cum_exp6}
& \E[\w{a_1a_2a_3}\w{a_4a_5a_6}] 
= \kappa( a_1 , a_4   )\kappa( a_2 , a_5)
\kappa( a_3 , a_6   ) + \kappa( a_1 , a_4 )\kappa( a_2 , a_6 )\kappa( a_3 , a_5 ) 
\nonumber \\ & \quad 
+ \kappa( a_1 , a_5)\kappa( a_2 , a_4)\kappa( a_3 , a_6) + 
\kappa( a_1 , a_5 )\kappa( a_2 , a_6 )\kappa( a_3 , a_4)
\nonumber \\ & \quad 
 + \kappa( a_1 , a_6 )\kappa( a_2 , a_4 )\kappa( a_3 , a_5) + \kappa( a_1 , a_6 
)\kappa( a_2 , a_5 )\kappa( a_3 ,a_4 ) 
\nonumber \\ & \quad 
 + \kappa(a_1,a_2,a_3,a_4,a_5,a_6) + \text{``}9 \times \kappa_2 \kappa_4\text{''}
\end{align}
where the last contribution denotes a sum of the nine terms consisting of a 
product of a second order cumulant and a fourth order cumulant.  Naturally, also some of the above
terms can be zero because of the gauge invariance constraints.

To better work with the expressions arising from (\ref{W_expansion}), let us next introduce a few shorthand notations.  We denote
\begin{eqnarray}
W(\mathbf{k}) &:=& W^{\lambda}_0(\mathbf{k})\, , \nonumber \\
\int d\mathbf{k}_{12 \cdots n}&:=& \int_{(\T^{d})^n} \! d\mathbf{k}_1 d\mathbf{k}_2 \cdots d\mathbf{k}_n 
\, , \nonumber \\
\delta(\mathbf{k} - \mathbf{k}_{ijk}) &:= & \delta(\mathbf{k} - \mathbf{k}_i - 
\mathbf{k}_j -\mathbf{k}_k)\, , \nonumber \\
\Omega^{++}_{--} &:=& \omega(\mathbf{k}_1) + \omega(\mathbf{k}_2) - 
\omega(\mathbf{k}_3) -\omega(\mathbf{k})\, , \nonumber \\
\Omega^{+-}_{-+} &:=& \omega(\mathbf{k}_1) - \omega(\mathbf{k}_2) - 
\omega(\mathbf{k}_3) +\omega(\mathbf{k}) \,.
\end{eqnarray}
We also choose $\sigma=1, \sigma'=-1$ 
and we will only consider the pairing contractions (i.e., the Gaussian 
contractions) in the expansion (\ref{cum_exp6}).  In fact, all terms arising from 
the non-pairing contractions are typically negligible in the kinetic scaling limit 
of the present type.
As an example, in Appendix \ref{sec:nonGauss} we show how the first order 
terms in (\ref{W_expansion}) vanish in the kinetic limit by assuming sufficient regularity 
of the dispersion relation $\omega$ and the $\ell_1$-clustering property of the 
fourth order cumulants.  As explained in \cite{LuSp11},
the contribution from the non-pairing terms in (\ref{cum_exp6}) can be controlled by similar techniques but we will
skip this more involved analysis here.

Hence, after integrating out the variables $\mathbf{k}_i$, $i=4,5,6$, the fourth term in (\ref{W_expansion}) gives
\begin{align}\label{eq:4thterm}
& 2\lambda^2 \delta(\mathbf{k}+\mathbf{k}') \bigg[
\int d\mathbf{k}_{123} \delta(\mathbf{k} - 
\mathbf{k}_{123}){W(\mathbf{k})W(\mathbf{k}_2)W(\mathbf{k}_3)}\int_0^t ds 
\int_0^s ds' \rme^{\ci(s-s')\Omega^{+-}_{-+}} \nonumber \\ 
& \quad - \int d\mathbf{k}_{123} \delta(\mathbf{k} - 
\mathbf{k}_{123}){W(\mathbf{k})W(-\mathbf{k}_1)W(\mathbf{k}_3)}\int_0^t ds 
\int_0^s ds' \rme^{\ci(s-s')\Omega^{+-}_{-+}} \nonumber \\ 
& \quad - \int d\mathbf{k}_{123} \delta(\mathbf{k} - 
\mathbf{k}_{123}){W(\mathbf{k})W(-\mathbf{k}_1)W(\mathbf{k}_2)}\int_0^t ds 
\int_0^s ds' \rme^{\ci(s-s')\Omega^{+-}_{-+}} \nonumber \\ 
& \quad +\int d\mathbf{k}_{123} \delta(\mathbf{k} - \mathbf{k}_{123}) 
{W(-\mathbf{k}_1)W(\mathbf{k}_2)W(\mathbf{k}_3)} \int_0^t ds \int_0^s ds' 
\rme^{\ci(s-s')\Omega^{+-}_{-+}} \bigg] + \NPC
\end{align}
where $\NPC$ stands for "non-pairing contraction terms". 
We proceed in the same way for the fifth term in (\ref{W_expansion}) yielding
\begin{align}\label{eq:5thterm}
& 2\lambda^2 \delta(\mathbf{k}+\mathbf{k}') \bigg[
 \int d\mathbf{k}_{123} \delta(\mathbf{k} + 
\mathbf{k}_{123}){W(\mathbf{k})W(-\mathbf{k}_1)W(-\mathbf{k}_2)}\int_0^t ds 
\int_0^s ds' \rme^{\ci(s-s')\Omega^{++}_{--}}  \nonumber \\ 
& \quad -\int d\mathbf{k}_{123} \delta(\mathbf{k} + \mathbf{k}_{123})
{W(\mathbf{k})W(-\mathbf{k}_1)W(\mathbf{k}_3)}\int_0^t ds 
\int_0^s ds' \rme^{\ci(s-s')\Omega^{++}_{--}}  \nonumber \\ 
& \quad - \int d\mathbf{k}_{123} \delta(\mathbf{k} + 
\mathbf{k}_{123}){W(\mathbf{k})W(-\mathbf{k}_2)W(\mathbf{k}_3)}\int_0^t ds 
\int_0^s ds' \rme^{\ci(s-s')\Omega^{++}_{--}}\nonumber \\ 
& \quad +\int d\mathbf{k}_{123} \delta(\mathbf{k} + \mathbf{k}_{123}) 
{W(-\mathbf{k}_1)W(-\mathbf{k}_2)W(\mathbf{k}_3)} \int_0^t ds \int_0^s ds' 
\rme^{\ci(s-s')\Omega^{++}_{--}} \bigg] +\NPC \,.
\end{align}
By changing integration variables so that $\mathbf{k}_1\to -\mathbf{k}_1$ in (\ref{eq:4thterm}) and 
$\mathbf{k}_1\to -\mathbf{k}_3$, $\mathbf{k}_2\to -\mathbf{k}_2$, $\mathbf{k}_3\to \mathbf{k}_1$ in (\ref{eq:5thterm}), we obtain
\begin{align}
& W^{\lambda}_t(\mathbf{k})- W(\mathbf{k})-(\pazocal{R}_3(\mathbf{k},t)+\NPC) 
\nonumber \\  & \quad 
= 2 \lambda^2 \int_{(\T^{d})^3} 
d\mathbf{k}_{123} \delta(\mathbf{k} + \mathbf{k}_{1} - 
\mathbf{k}_{2}-\mathbf{k}_{3}) \int_0^t ds \int_0^s ds' 
\rme^{\ci(s-s')(\omega_1-\omega_2-\omega_3+\omega)}
\nonumber \\  & \qquad\quad
\times [W(\mathbf{k})W(\mathbf{k}_2)W(\mathbf{k}_3)- 
W(\mathbf{k})W(\mathbf{k}_1)W(\mathbf{k}_3)- 
W(\mathbf{k})W(\mathbf{k}_1)W(\mathbf{k}_2)+ 
W(\mathbf{k}_1)W(\mathbf{k}_2)W(\mathbf{k}_3)]\nonumber \\ 
& \qquad +2 \lambda^2 \int d\mathbf{k}_{123} \delta(\mathbf{k} + 
\mathbf{k}_{1} - \mathbf{k}_{2}-\mathbf{k}_{3}) \int_0^t ds \int_0^s ds' 
\rme^{-\ci(s-s')(\omega_1-\omega_2-\omega_3+\omega)}
\nonumber \\  & \qquad\quad
\times [W(\mathbf{k})W(\mathbf{k}_2)W(\mathbf{k}_3) 
- W(\mathbf{k})W(\mathbf{k}_1)W(\mathbf{k}_3)
- W(\mathbf{k})W(\mathbf{k}_1)W(\mathbf{k}_2)
+ W(\mathbf{k}_1)W(\mathbf{k}_2)W(\mathbf{k}_3) ]
\nonumber \\  & \quad
=  2 \lambda^2 \int d\mathbf{k}_{123} 
\delta(\mathbf{k} + \mathbf{k}_{1} - \mathbf{k}_{2}-\mathbf{k}_{3}) \int_{0}^t 
ds \int_{-s}^s d r\, \rme^{\ci r(\omega_1-\omega_2-\omega_3+\omega)}
\nonumber \\  & \qquad
 \times [
W(\mathbf{k}_1)W(\mathbf{k}_2)W(\mathbf{k}_3) 
+ W(\mathbf{k})W(\mathbf{k}_2)W(\mathbf{k}_3) - 
W(\mathbf{k})W(\mathbf{k}_1)W(\mathbf{k}_3) - 
W(\mathbf{k})W(\mathbf{k}_1)W(\mathbf{k}_2)] 
\label{final_sec_ord} \,.
\end{align}
Note that for any $\Omega \in \R$ we have 
$\lambda^2\int_0^t ds \int_{|r|\le s} d r\, \rme^{\ci r\Omega} =
 \int_{|r|\le t} d r\, \rme^{\ci r\Omega} (\lambda^2 t- \lambda^2 |r|)$.  By setting $t=\tau \lambda^{-2}$ and taking
$\lambda \to 0$, this expression formally converges to 
$\tau \int_{-\infty}^\infty d r\, \rme^{\ci r\Omega} = \tau 2 \pi \delta(\Omega) $.  Therefore, doing this 
in (\ref{final_sec_ord}) yields the conjecture that 
\begin{align}\label{eq:Wtau}
& {W}_{\tau}(\mathbf{k}) - W(\mathbf{k}) -O(\tau^2)  = 
 \tau  4 \pi \int\!  d\mathbf{k}_{123} 
\delta(\mathbf{k} + \mathbf{k}_{1} - \mathbf{k}_{2}-\mathbf{k}_{3}) 
\delta(\omega + \omega_{1} - \omega_{2}-\omega_{3}) \nonumber \\ 
& \quad  \times [
W(\mathbf{k}_1)W(\mathbf{k}_2)W(\mathbf{k}_3) 
+ W(\mathbf{k})W(\mathbf{k}_2)W(\mathbf{k}_3) - 
W(\mathbf{k})W(\mathbf{k}_1)W(\mathbf{k}_3) - 
W(\mathbf{k})W(\mathbf{k}_1)W(\mathbf{k}_2)] \, .
\end{align}
Since here $W(\mathbf{k})=W_0(\mathbf{k})$, if we divide the left hand side by $\tau$ and then take $\tau\to 0$, it converges to 
$\partial_{\tau} {W}_{\tau}(\mathbf{k})$ at $\tau=0$.  Dividing the right hand side of (\ref{eq:Wtau}) by $\tau$ yields
$\pazocal{C}(W_0(\cdot))(\mathbf{k})$, as defined in (\ref{eq:defCBP}).  Therefore, the Boltzmann-Peierls equation 
(\ref{Boltzmann}) should hold at $\tau=0$.  Assuming that the state of the original 
system remains so regular that the estimates leading to the conjecture continue to hold, we thus find that the Boltzmann-Peierls equation
should be valid for the limit of $W^{\lambda}_{\tau\lambda^{-2}}$ also at later times $\tau$, as was claimed in the beginning 
of the subsection.

\subsubsection{Decay of field time-correlations}\label{sec:timecorr}

As a second example of how the standard perturbation expansion 
works for the cumulants,
we consider a kinetic scaling limit of time-correlations.  
In particular, our goal is to show how the main results proven in \cite{LuSp11}
relate to the present cumulant hierarchy expansions.  

The notation 
``$\hat{a}_t$'' was used in \cite{LuSp11} to define the finite periodic lattice analogue of the present 
$a_t$-field. (One can compare the definition of ``$\hat{a}_t$''
in (3.9) and its evolution equation in (3.10) in \cite{LuSp11} to those given for $a_t$ here.)
Translated to the present infinite lattice setup,  the main theorem of \cite{LuSp11} (Theorem 2.4)
leads to the following conjecture about the  decay of time correlations of $a_t$:
start the system from an $\ell_1$-clustering equilibrium Gibbs state.
Then there is a continuous function $A^\lambda_t(\mathbf{k})$ such that 
$\E[a_0(\mathbf{k}',-1)a_t(\mathbf{k},1)]=\delta(\mathbf{k}'+\mathbf{k})A^\lambda_t(\mathbf{k})$.
The conjecture is that the kinetic scaling limit of $A^\lambda$ exists and its decay is governed by 
the ``loss term'' of the Boltzmann-Peierls equation (\ref{Boltzmann}) evaluated at the corresponding limit equilibrium
covariance function $\Weql(\mathbf{k})=\beta^{-1}/(\omega(\mathbf{k})-\mu)$ where $\beta>0$ and $\mu\in \R$ are parameters determined 
by the equilibrium state.  (Such functions $\Weql$ are indeed stationary solutions of (\ref{Boltzmann}).)
More precisely,  \cite[Theorem 2.4]{LuSp11} is consistent with the conjecture that
\begin{align}\label{eq:Aeqlconj}
& \lim_{\lambda\to 0} A^\lambda_{\tau\lambda^{-2}}(\mathbf{k}) = \Weql(\mathbf{k}) \rme^{-\tau \Gamma(\Weql(\cdot))(\mathbf{k})}\, ,
\end{align}
where
\begin{align}\label{eq:defGamma}
& \Gamma(W(\cdot))(\mathbf{k}) = -2 \int_0^\infty \!\rmd r
\int_{(\T^{d})^3} d\mathbf{k}_1d\mathbf{k}_2 d\mathbf{k}_3 \delta(\mathbf{k} + \mathbf{k}_{1} - \mathbf{k}_{2}-\mathbf{k}_{3}) 
 \rme^{\ci r(\omega_1-\omega_2-\omega_3+\omega)} \nonumber \\ 
& \quad  \times [
W(\mathbf{k}_2)W(\mathbf{k}_3) 
-W(\mathbf{k}_1)W(\mathbf{k}_3) 
-W(\mathbf{k}_1)W(\mathbf{k}_2)] \, .
\end{align}

Instead of assuming that the system starts from an equilibrium state, let us consider more general 
states which we assume to be gauge and translation invariant and $\ell_1$-clustering.  
We can immediately use the results derived in Section \ref{sec:cwdyn} if we consider the ``$a_0$'' term to be a new field which
has trivial time evolution with zero amplitudes, i.e., the corresponding $\mathcal{I}_j$-set is empty.  The net effect of this change 
is that more than half of the terms analyzed in the previous section will be absent.  For instance,
the expansion of the time correlation using (\ref{wick_evolution1}) reads  
\begin{align}\label{time_corr_evol}
& \kappa[a_0(\mathbf{k}',-1), a_t(\mathbf{k},1)]
 \nonumber \\  & \quad 
 = \kappa[a_0(\mathbf{k}',-1), a_0(\mathbf{k},1)] 
 - \ci\lambda \int_0^t ds\int \! d\mathbf{k}_{123} \delta(\mathbf{k} - \mathbf{k}_{123})\rme^{\ci s(\omega + 
\omega_1 -\omega_2 -\omega_3)} \kappa[(a_0)_{I'+(\mathbf{k}',-1)}] \nonumber \\ 
& \qquad
-\lambda^2 \sum_{\ell \in I'} \sigma_\ell \int_0^t ds \int  
d\mathbf{k}_1 d\mathbf{k}_2 d\mathbf{k}_3 \delta(\mathbf{k} - \mathbf{k}_1 - 
\mathbf{k}_2 -\mathbf{k}_3)\rme^{\ci s(\omega +\omega_1 - \omega_2 
- \omega_3)} \nonumber \\ 
& \qquad \quad\times \int_0^s ds'   \int  d\mathbf{k}_4 d\mathbf{k}_5 
d\mathbf{k}_6 \delta(\mathbf{k}_{\ell} - \mathbf{k}_4 - \mathbf{k}_5 
-\mathbf{k}_6)\rme^{\ci s'(\sigma_\ell \omega_\ell +\omega_4 - \sigma_\ell \omega_5 - 
\omega_6)} \E \big[ \w{a_{s'}^{J_{\ell}}} \w{a_0(\mathbf{k}',-1) a_{s'}^{\cancelm{\ell}{I'} } } \big] 
\end{align}
where $I'=((\mathbf{k}_1,-1), (\mathbf{k}_2,1), (\mathbf{k}_3,1) )$ and 
$J_{\ell}=((\mathbf{k}_4,-1), (\mathbf{k}_5,\sigma_{\ell}), (\mathbf{k}_6,1)) $.

As in Section \ref{derivation_boltzmann}, we now assume that only pairings contribute in the kinetic scaling limit.
When applying (\ref{cum_exp6}) to expand $\E[ \w{a_{s'}^{J_{\ell}}} \w{a_0(\mathbf{k}',-1) a_{s'}^{\cancelm{\ell}{I'} } } ]$,
we note that every pairing term results in a product containing a factor $A^\lambda_{s'}(-\mathbf{k}')$ and a product of two $W_{s'}^\lambda$-terms.
The rest of the structure is identical to the one considered earlier, some of the terms are merely missing now.
We then use the perturbation expansion to the product once more.  This produces a term where $s'$ is set to $0$, and a remainder which we
assume to be negligible as before.  The rest of the computation is essentially the same as in Section \ref{derivation_boltzmann},
yielding
\begin{align}\label{half_boltzmann}
& A^\lambda_{\tau\lambda^{-2}}(\mathbf{k})- A^\lambda_0(\mathbf{k})- (\text{terms higher order in }\lambda\text{ or }\tau)
\nonumber \\ & \quad 
= 2 \lambda^2 \int_{(\T^{d})^3} 
d\mathbf{k}_{123} \delta(\mathbf{k} + \mathbf{k}_{1} - 
\mathbf{k}_{2}-\mathbf{k}_{3}) 
\int_0^{\tau\lambda^{-2}}\!\! ds \int_0^s d r\, \rme^{\ci r(\omega_1-\omega_2-\omega_3+\omega)}
\nonumber \\  & \qquad
\times [A^\lambda_0(\mathbf{k})W^\lambda_0(\mathbf{k}_2)W^\lambda_0(\mathbf{k}_3)- 
A^\lambda_0(\mathbf{k})W^\lambda_0(\mathbf{k}_1)W^\lambda_0(\mathbf{k}_3)- 
A^\lambda_0(\mathbf{k})W^\lambda_0(\mathbf{k}_1)W^\lambda_0(\mathbf{k}_2)]\,.
\end{align}
Hence, if we divide the equation by $\tau$ and then 
take $\lambda\to 0$, followed by $\tau\to 0$, we find that $A_\tau(\mathbf{k}) :=\lim_{\lambda\to 0} A^\lambda_{\tau\lambda^{-2}}(\mathbf{k})$
should satisfy at $\tau=0$
\begin{align}\label{eq:Aevol}
& \partial_\tau A_\tau(\mathbf{k}) = -A_\tau(\mathbf{k}) \Gamma(W_\tau(\cdot))(\mathbf{k})\, ,
\end{align}
where $\Gamma$ has been defined in (\ref{eq:defGamma}).  As before, the conjecture is that this equation continues to hold for other values $\tau>0$, as well.

Once $W_\tau$ is given, equation (\ref{eq:Aevol}) is straightforward to solve.  Since $A^\lambda_0(\mathbf{k})=W_0(\mathbf{k})$, 
the solution reads
\begin{align}\label{eq:Atausol}
& A_\tau(\mathbf{k}) = W_0(\mathbf{k}) \rme^{-\int_0^\tau\!\rmd s\, \Gamma(W_s(\cdot))(\mathbf{k})}\, .
\end{align}
If the system is started in an equilibrium state with $W_0=\Weql$, we have $W_s=\Weql$ for all $s$.
Thus (\ref{eq:Atausol}) implies (\ref{eq:Aeqlconj}) in this special case.

\section{Discussion about further applications}\label{sec:discussion}

\subsection{Limitations of the direct renormalization procedure: inhomogeneous DNLS}\label{sec:inhomog}

The field renormalization used with the translational invariant data
greatly simplified the evolution equation by removing the linear term.  
The renormalization procedure, given in (\ref{eq:defatfield}), was a simple multiplication by a 
time and $k$-dependent function and the time-dependent first order terms had no effect in the interaction term.
Unfortunately, this case is atypical: most commonly, the necessary renormalization is not a multiplication operator
and the first order terms will also affect the oscillatory phase terms arising from the harmonic evolution.
In fact, this happens also for the DNLS model as soon as we drop the requirement that the initial data is translation 
invariant.  To explain the changes needed in the renormalization procedure, we discuss in this subsection the DNLS model with
inhomogeneous initial data in some more detail.

Before considering the inhomogeneous case, let us begin with an example which 
emphasizes the importance of the field renormalization even for translation invariant initial data
if one considers taking kinetic scaling limits of all field observables.  We
inspect the time correlation of the ``bare'' $ \FT{\psi}$-fields, i.e., 
$\E[\FT{\psi}_0(\mathbf{k}',-1)\FT{\psi}_t(\mathbf{k},1)]$ 
assuming spatially homogeneous, gauge invariant initial data. Then there exists a function $ \Psi_t(\mathbf{k})$ such that  
$ \E[\FT{\psi}_0(\mathbf{k}',-1)\FT{\psi}_t(\mathbf{k},1)]= 
\delta(\mathbf{k}'+\mathbf{k})\Psi_t(\mathbf{k})$.  
By using (\ref{eq:NLStiFT}) in (\ref{cum_evolution}), we find the following evolution equation for $\Psi_t$:
\begin{align}\label{psi_first_order}
& 
\delta(\mathbf{k}' + \mathbf{k})\Psi_t(\mathbf{k})
= \kappa(\FT{\psi}_0(\mathbf{k}',-1),\FT{\psi}_t(\mathbf{k},1)) 
=\delta(\mathbf{k}' + \mathbf{k})\Psi_0(\mathbf{k})
   -\ci  \delta(\mathbf{k}'+\mathbf{k})
   \int_0^t\!ds\,\omega_s^{\lambda}(\mathbf{k})\Psi_s(\mathbf{k}) \nonumber 
\\ \quad
&- \ci \lambda \int_0^t ds\int_{\T^{3d}} d\mathbf{k}_2 d\mathbf{k}_3 
d\mathbf{k}_4 \delta(\mathbf{k} - \mathbf{k}_1 - \mathbf{k}_2 -\mathbf{k}_3) 
\E[\w{\FT{\psi}_0(\mathbf{k}',-1)} 
\w{\FT{\psi}_t(\mathbf{k}_1,-1)\FT{\psi}_t(\mathbf{k}_2,1)\FT{\psi}
_t(\mathbf{k}_3,1)}]\,.
\end{align}

The linear equation associated to (\ref{psi_first_order}) thus has the form
\begin{equation}
f_t(\mathbf{k})= f_0(\mathbf{k})
   -\ci \int_0^t\!ds\,\omega_s^{\lambda}(\mathbf{k})f_s(\mathbf{k}) \, ,
\label{renormalizanda_eq}
\end{equation}
which is solved by
$f_t(\mathbf{k})=U_t(\mathbf{k})f_0(\mathbf{k})$
where $ U_t(\mathbf{k})=\exp(-\ci\int_0^tds\, \omega^{\lambda}_s(\mathbf{k}))$. We recall that $ 
\omega^{\lambda}_s(\mathbf{k}) = \omega(\mathbf{k}) + \lambda R_s$, and thus
at a kinetic time scale, with $t=\tau \lambda^{-2}$, we have
$\int_0^t\! ds\, \omega^{\lambda}_s(\mathbf{k}) = \tau \lambda^{-2} \omega(\mathbf{k}) + O(\lambda^{-1})$.
Therefore, $U_t$ has unbounded oscillations in the kinetic scaling limit.  Also, we find that even though the effect of
the first order term proportional to $R_s$ is subdominant, it is still rapidly oscillating on the kinetic time-scale and should not
be ``expanded'' in any perturbative treatment of the problem.

We could solve the problem with these unbounded oscillations by considering instead of 
$\FT{\psi}_t$ the renormalized field $a_t=U_t^{-1}\FT{\psi}_t$.
In fact, by the results of Section \ref{sec:timecorr} and using $a_0=\FT{\psi}_0$ we find that
\begin{align}
& \delta(\mathbf{k}+\mathbf{k}')A^\lambda_t(\mathbf{k})
=
\E[a_0(\mathbf{k}',-1)a_t(\mathbf{k},1)] 
= U_t(\mathbf{k}) \E\!\left[\FT{\psi}_0(\mathbf{k}',-1)\FT{\psi}_t(\mathbf{k},1)\right]
= \delta(\mathbf{k}+\mathbf{k}') U_t(\mathbf{k})  \Psi_t(\mathbf{k}) \, .
\end{align}
We have argued in Section \ref{sec:timecorr} that the kinetic scaling limit of $A^\lambda_t$ exists.
Then $\Psi_{\tau \lambda^{-2}}(\mathbf{k})$ \defem{cannot} have a convergent limit as $\lambda\to 0$;
instead, it has fast oscillations proportional to $U^{-1}_{\tau \lambda^{-2}}(\mathbf{k})$.
Let us also once more stress that the ``zeroth order renormalization'', i.e., countering the free evolution,
does not remove all of the unbounded oscillations but still leaves those resulting from the $R_t$ term.

However, the above renormalization procedure cannot be straightforwardly extended to more complicated cases.
Consider next the DNLS model with inhomogeneous initial data.  As in section 
\ref{derivation_boltzmann}, our goal is to find the right observable which 
satisfies the Boltzmann equation in the kinetic scaling limit. The evolution equation 
for the bare field $ \FT{\psi}$ reads
\begin{align}\label{psi_inhomo}
&\partial_t \FT{\psi}_t(\mathbf{k},\sigma)= -\ci\sigma 
\omega(\mathbf{k})\FT{\psi}_t(\mathbf{k},\sigma) \nonumber \\ 
& \quad -2\ci\sigma\lambda \int_{(\T^{d})^3} d\mathbf{k}_1 d\mathbf{k}_2 d\mathbf{k}_3 
\delta(\mathbf{k} - \mathbf{k}_1 - \mathbf{k}_2 -\mathbf{k}_3)
\kappa(\FT{\psi}_t(\mathbf{k}_1,-1), \FT{\psi}_t(\mathbf{k}_3,1)) 
\FT{\psi}_t(\mathbf{k}_2,\sigma)
\nonumber \\ & \quad 
-\ci\sigma\lambda\int_{(\T^{d})^3} d\mathbf{k}_1 d\mathbf{k}_2 d\mathbf{k}_3 
\delta(\mathbf{k} - \mathbf{k}_1 - \mathbf{k}_2 -\mathbf{k}_3)
\w{\FT{\psi}_t(\mathbf{k}_1,-1)\FT{\psi}_t(\mathbf{k}_2,\sigma)\FT{\psi}
_t(\mathbf{k}_3,1)} \, .
\end{align}
By following the same strategy as for the homogeneous case, let be 
$U_t$, $t\ge 0$, denote the family of linear operators which solves the linear 
part of 
(\ref{psi_inhomo}): we suppose that it solves the operator equation
$\partial_t U_t = - \ci H_t  U_t$ where 
$(H_t f)(\mathbf{k},\sigma) = \sigma \omega(\mathbf{k}) f(\mathbf{k},\sigma)
 +2\sigma\lambda \int_{\T^{d}} d\mathbf{k}' K_t(\mathbf{k} -\mathbf{k}') 
f(\mathbf{k}',\sigma)$.
The time-dependent convolution kernel $K_t(\mathbf{k})$ should be equal 
to 
$\int_{\T^{d}} d\mathbf{k}_1 \kappa(\FT{\psi}_t(\mathbf{k}_1,-1), 
\FT{\psi}_t(\mathbf{k}- \mathbf{k}_1,1))$.  
This can be done either by first solving the implicit equation 
for the above integral over the cumulant, or by leaving $K_t$ arbitrary and 
fixing it by some minimization procedure at the end.  
If such a family $U_t$ can be found, we may define as before 
$\tilde{a}_t(\mathbf{k}, \sigma)=  (U_{t}^{-1}\FT{\psi}_t)(\mathbf{k},\sigma)$ 
and, since 
$\partial_t U_{t}^{-1}= -U_{t}^{-1}( \partial_t U_{t} )U_{t}^{-1}=\ci U_{t}^{-1} 
H_t$, 
it then satisfies an evolution equation 
\begin{align}\label{at_inhomo}
&\partial_t \tilde{a}_t (\mathbf{k},\sigma) =
-\ci\sigma\lambda \int_{(\T^{d})^3} d^3\mathbf{k} \int_{\T^{d}} 
d\mathbf{k}''\int_{(\T^{d})^3} d^3\mathbf{k}'
\delta(\mathbf{k}'' - \mathbf{k}'_1 - \mathbf{k}'_2 -\mathbf{k}'_3)
\nonumber \\ & \qquad \times
\tilde{u}_t(\mathbf{k},\mathbf{k}'',\sigma)
u_t(\mathbf{k}'_1,\mathbf{k}_1,-1)
u_t(\mathbf{k}'_2,\mathbf{k}_2,\sigma)
u_t(\mathbf{k}'_3,\mathbf{k}_3,1)
\w{\,\tilde{a}_t(\mathbf{k}_1,-1)\tilde{a}_t(\mathbf{k}_2,\sigma)\tilde{a}
_t(\mathbf{k}_3,1)}
\, ,
\end{align}
where $u_t(\mathbf{k},\mathbf{k}',\sigma)$ and 
$\tilde{u}_t(\mathbf{k},\mathbf{k}',\sigma)$ 
denote the formal integral kernels of the operators $U_t$ and $U_t^{-1}$, 
respectively.
(Note that the operators $U_t$ are diagonal in $\sigma$ but not any more in $ 
\mathbf{k}$.)

Apart from some special cases it seems difficult to gain sufficient control over the 
operators 
$U_t$ to consider taking a kinetic limit using the observables $\tilde{a}_t$, 
unlike with the explicit phases factors which appeared in the spatially homogeneous 
case.
Even though $U_t$ approach the same multiplication operator as before when 
$\lambda \to 0$ for a fixed $t$, it is not clear 
that the corrections do not contribute in the limit, since we need to consider 
$t=O(\lambda^{-2})$.
Thus, although the cumulant expansion of $\tilde{a}_t$-fields  
is simpler than that of $\psi_t$-fields, to control the kinetic scaling limit 
looks
intractable.  Hence, new approaches for the study of the kinetic time scales are
called for.

\subsection{Kinetic theory beyond kinetic time-scales?}\label{sec:beyondkin}

In this section we propose a new approach to the problem when a 
renormalization scheme with a convergent kinetic scaling limit cannot be found 
or controlled.  The approach does not require taking $\lambda\to 0$, and, if 
successful, it may also yield estimates which are valid beyond the standard
kinetic time scales which was $O(\lambda^{-2})$ in the above DNLS case.  

The main idea can be summarized in the following simple
observation.  Suppose $f_t$ is a solution to the equation
\begin{align}\label{eq:startf}
 f_t = f_0 + R_t + \int_0^t \! \rmd s\, F_{s,t}[f_s]
\end{align}
where $R_t=O(\vep)$ uniformly in $t$ and 
$F_{s,t}$, $0\le s\le t$, is an explicit, but possibly nonlinear functional of 
$f_s$.
Suppose furthermore that there is another, ``simpler'', functional $\Phi_{s,t}$ 
such that $F_{s,t} = \Phi_{s,t} + O(\vep (1+|t-s|)^{-p})$ with $p>1$;
by simpler we mean that the evolution problem
\begin{align}\label{eq:newg}
 \varphi_t = S_t + \int_0^t \! \rmd s\, \Phi_{s,t}[\varphi_s] \, ,
\end{align}
for any bounded ``source term'' $S_t$, is easier to study than (\ref{eq:startf}).
Under these assumptions, any solution to (\ref{eq:startf}) satisfies
\begin{align}\label{eq:newf}
 f_t = f_0 + \rho_t + \int_0^t \! \rmd s\, \Phi_{s,t}[f_s] \, , 
\end{align}
where $\rho_t := R_t + \int_0^t \! \rmd s\, (F_{s,t}[f_s]-\Phi_{s,t}[f_s]) $ is 
$O(\vep)$ uniformly in $t$.
Therefore, if we could prove that (\ref{eq:newg}) is stable under 
perturbations of the source term $S_t$, we may conclude that 
the solution $\varphi_t$ to 
\begin{align}
 \varphi_t = f_0 + \int_0^t \! \rmd s\, \Phi_{s,t}[\varphi_s] 
\end{align}
then approximates the ``true'' solution $f_t$ with an error which is $O(\vep)$ 
\defem{uniformly in time}.

To have a concrete example, consider again the DNLS with gauge invariant initial data.  
We sketch below two conjectures about this system.  
The details of the conjectures should not be taken too seriously: they should be considered more as
examples of what could happen in general rather than as specific conjectures about the behavior of 
the DNLS system.  For this reason, the discussion will be kept on a very loose level; in particular, we do
not wish to make any specific claims about what kind of metrics should be used for studying the uniform boundedness in time.

Let us begin with the case of inhomogeneous, gauge invariant initial data.  
We move to slowly varying fields by cancelling the free evolution term.  This 
renormalization leads to equations which are
almost identical to those in the homogeneous case.  Namely, for the field
$b_t(\mathbf{k},\sigma):=e^{\ci\sigma\omega(\mathbf{k})t}\hat{\psi}_t(\mathbf{k},\sigma)$ 
we have
\begin{align}\label{eq:btevoleq}
&\partial_t b_t(\mathbf{k},\sigma)= 
-2\ci\sigma\lambda \int_{(\T^{d})^3} d\mathbf{k}_1 d\mathbf{k}_2 d\mathbf{k}_3 
\delta(\mathbf{k} - \mathbf{k}_1 - \mathbf{k}_2 -\mathbf{k}_3)
\rme^{\ci t(\sigma\omega + \omega_1 -\sigma\omega_2 -\omega_3)} 
B_t(\mathbf{k}_1,\mathbf{k}_3) \w{b_t(\mathbf{k}_2,\sigma)}
\nonumber \\  & \quad 
-\ci\sigma\lambda\int_{(\T^{d})^3} d\mathbf{k}_1 d\mathbf{k}_2 d\mathbf{k}_3 
\delta(\mathbf{k} - \mathbf{k}_1 - \mathbf{k}_2 -\mathbf{k}_3)
\rme^{\ci t(\sigma \omega + \omega_1 -\sigma \omega_2 -\omega_3)} 
\w{b_t(\mathbf{k}_1,-1)b_t(\mathbf{k}_2,\sigma)b_t(\mathbf{k}_3,1)} \,
\end{align}
where $B_t(\mathbf{k}',\mathbf{k})=\kappa(b_t(\mathbf{k}',-1), 
b_t(\mathbf{k},1))$.  By the gauge invariance, 
$\kappa(b_t(\mathbf{k}',\sigma'), b_t(\mathbf{k},\sigma))=0$, unless 
$\sigma'+\sigma=0$.  Hence, there is only one other
nonzero second order cumulant,
$\kappa(b_t(\mathbf{k}',1), 
b_t(\mathbf{k},-1))=B_t(\mathbf{k},\mathbf{k}')=B_t(-\mathbf{k}',-\mathbf{k}
)^*$.  

To study the fourth order cumulants, it suffices to concentrate on the function
$$
D_t(k) := 
\kappa(b_t(\mathbf{k}_1,-1),b_t(\mathbf{k}_2,-1),b_t(\mathbf{k}_3,1),b_t(\mathbf
{k}_4,1)) \,, \, k\in (\T^d)^4 \,. 
$$
Then the first two equations in the cumulant hierarchy are equivalent to
\begin{align}\label{eq:inhomogBt}
& \partial_t B_t(\mathbf{k}',\mathbf{k}) = 
-2\ci\lambda \int_{(\T^{d})^3} d\mathbf{k}_1 d\mathbf{k}_2 d\mathbf{k}_3 
\delta(\mathbf{k} - \mathbf{k}_1 - \mathbf{k}_2 -\mathbf{k}_3)
\rme^{\ci t(\omega + \omega_1 -\omega_2 -\omega_3)} 
B_t(\mathbf{k}_1,\mathbf{k}_3) B_t(\mathbf{k}_2,\mathbf{k}')
 \nonumber \\ & \quad 
+2\ci\lambda \int_{(\T^{d})^3} d\mathbf{k}_1 d\mathbf{k}_2 d\mathbf{k}_3 
\delta(\mathbf{k}' - \mathbf{k}_1 - \mathbf{k}_2 -\mathbf{k}_3)
\rme^{\ci t(-\omega' + \omega_1 +\omega_2 -\omega_3)} 
B_t(\mathbf{k}_1,\mathbf{k}_3) B_t(\mathbf{k},\mathbf{k}_2)
 \nonumber \\ & \quad 
-\ci\lambda\int_{(\T^{d})^3} d\mathbf{k}_1 d\mathbf{k}_2 d\mathbf{k}_3 
\delta(\mathbf{k} - \mathbf{k}_1 - \mathbf{k}_2 -\mathbf{k}_3)
\rme^{\ci t(\omega + \omega_1 -\omega_2 -\omega_3)} 
D_t(\mathbf{k}',\mathbf{k}_1,\mathbf{k}_2,\mathbf{k}_3)
 \nonumber \\ & \quad 
+\ci\lambda\int_{(\T^{d})^3} d\mathbf{k}_1 d\mathbf{k}_2 d\mathbf{k}_3 
\delta(\mathbf{k}' - \mathbf{k}_1 - \mathbf{k}_2 -\mathbf{k}_3)
\rme^{\ci t(-\omega' + \omega_1 + \omega_2 -\omega_3)} 
D_t(\mathbf{k}_1,\mathbf{k}_2,\mathbf{k}_3,\mathbf{k}) \, ,
\\
& \partial_t D_t(k) = 
 -\ci 2 \lambda \sum_{\ell=1}^4 \sigma_\ell
 \int_{(\T^{d})^3} d\mathbf{k}'_1 d\mathbf{k}'_2 d\mathbf{k}'_3 
\delta(\mathbf{k}_\ell - \mathbf{k}'_1 - \mathbf{k}'_2 -\mathbf{k}'_3)
\rme^{\ci t(\sigma_\ell\omega_\ell + \omega'_1 -\sigma_\ell\omega'_2 
-\omega'_3)}
 \nonumber \\ & \qquad \times
B_t(\mathbf{k}'_1,\mathbf{k}'_3) D_t(\text{``replace }\mathbf{k}_\ell\text{ by } 
\mathbf{k}'_2\text{''})
 \nonumber \\ & \quad 
 -\ci \lambda \sum_{\ell=1}^4 \sigma_\ell
 \int_{(\T^{d})^3} d\mathbf{k}'_1 d\mathbf{k}'_2 d\mathbf{k}'_3 
\delta(\mathbf{k}_\ell - \mathbf{k}'_1 - \mathbf{k}'_2 -\mathbf{k}'_3)
\rme^{\ci t(\sigma_\ell\omega_\ell + \omega'_1 -\sigma_\ell\omega'_2 
-\omega'_3)}
 \nonumber \\ & \qquad \times
 (\text{``}(1\times \kappa_6) + (5\times B_t D_t) + (6\times B_t B_t 
B_t)\text{''})
 \label{eq:inhomogGt}
\, ,
\end{align}
where in the second formula, $\sigma := (-1,-1,1,1)$ and
on the last line we have applied (\ref{cum_exp6}) and 
merely denoted how many nonzero terms each type of 
partition can have.
Let us point out that
the first two terms in (\ref{eq:inhomogBt}) cancel each other out if the state 
is spatially homogeneous, since then 
$B_t(\mathbf{k}',\mathbf{k}) \propto \delta(\mathbf{k}'+\mathbf{k})$.  For an 
inhomogeneous state, however, the cancellation need not be exact, and since
this term is then not $O(\lambda^2)$, it will likely prevent taking of the kinetic scaling 
limit of $B_t$ directly.

If we now integrate the $B$ equation as in (\ref{cum_expansion}), we find that
\begin{align}\label{eq:Btfull}
& B_t(\mathbf{k}',\mathbf{k}) = B_0(\mathbf{k}',\mathbf{k})-\ci \lambda 
\int_0^t\,\rmd s\, 
 \int_{(\T^{d})^3} d\mathbf{k}_1 d\mathbf{k}_2 d\mathbf{k}_3 \delta(\mathbf{k} - 
\mathbf{k}_1 - \mathbf{k}_2 -\mathbf{k}_3)
\rme^{\ci s(\omega + \omega_1 -\omega_2 -\omega_3)} 
 \nonumber \\ & \qquad \times
\left( 2 B_s(\mathbf{k}_1,\mathbf{k}_3) B_s(\mathbf{k}_2,\mathbf{k}') 
+D_0(\mathbf{k}',\mathbf{k}_1,\mathbf{k}_2,\mathbf{k}_3)\right)
 \nonumber \\ & \quad 
+\ci \lambda \int_0^t\,\rmd s\,
\int_{(\T^{d})^3} d\mathbf{k}_1 d\mathbf{k}_2 d\mathbf{k}_3 \delta(\mathbf{k}' - 
\mathbf{k}_1 - \mathbf{k}_2 -\mathbf{k}_3)
\rme^{\ci s(-\omega' + \omega_1 +\omega_2 -\omega_3)} 
 \nonumber \\ & \qquad \times
\left( 2 B_s(\mathbf{k}_1,\mathbf{k}_3) B_s(\mathbf{k},\mathbf{k}_2) +
D_0(\mathbf{k}_1,\mathbf{k}_2,\mathbf{k}_3,\mathbf{k})\right)
 \nonumber \\ & \quad 
+ \lambda^2 \int_0^t\,\rmd s' \int_{s'}^t\,\rmd s\, \sum \left(\ \cdots\right)
\, .
\end{align}
In the final sum, each term depends on $s$ only via the oscillatory phases.  
These can be collected together and they have a structure
\begin{align}
  \rme^{\ci s(\sigma_0(\omega_0-\omega_2) + \omega_1 -\omega_3)} 
 \rme^{\ci s' (\tilde\sigma_\ell\tilde\omega_\ell + \omega'_1 
-\tilde\sigma_\ell\omega'_2 -\omega'_3)} \, ,
\end{align}
where $(\sigma_0,\mathbf{k}_0)$ is equal to $(1,\mathbf{k})$, if the term arises 
from the second  last term in (\ref{eq:inhomogBt}) and 
it is equal $(-1,\mathbf{k}')$ if it arises from the last term.  The pair
$(\tilde\sigma_\ell,\tilde\omega_\ell)$ comes from arguments of the 
corresponding ``$D_t$-term'' and thus depends also on this choice.  Therefore, 
the $s$-integral over the oscillatory phase 
can be computed explicitly and to each of the terms in the sum it will produce a 
factor
\begin{align} \label{eq:integratephase}
 \rme^{\ci s' (\sigma_0(\omega_0-\omega_2) + \omega_1 
-\omega_3+\tilde\sigma_\ell\tilde\omega_\ell + \omega'_1 -\sigma_\ell\omega'_2 
-\omega'_3)}
\int_0^{t-s'}\!\!\rmd r \,  \rme^{\ci r (\sigma_0(\omega_0-\omega_2) + \omega_1 
-\omega_3)} 
 \, .
\end{align}

It is difficult to go further in the analysis of the oscillatory phases 
without resorting to graph theory,  and we will not pursue it here.
However, already the simple example given in Appendix \ref{sec:nonGauss} shows 
that, if the state is $\ell_1$-clustering, the 
oscillations may result in time-integrals which are absolutely convergent over 
$[0,\infty)$.  For instance, this explicit example implies that
if the initial state is homogeneous and $\ell_1$-clustering, then the two terms 
depending on $D_0$ in (\ref{eq:Btfull}) are $O(\lambda)$
\defem{uniformly in time}.  If we assume that similar bounds are valid for inhomogeneous
states and every term containing $\kappa_6$, the two-component field $f_t=(B_t,D_t)$
behaves as the model considered in the 
beginning of the section.  

Therefore, assuming uniform boundedness of $\kappa_6$ allows using the principles
described in the beginning of this section and
results in a conjecture about the evolution of the cumulants.
Suppose that there is a metric for the cumulants, similar or given by 
$\ell_1$-clustering estimates, such that the following results hold for
sufficiently regular dispersion relations $\omega$ and initial data:
\begin{enumerate}
 \item Suppose that the contribution from $\kappa_6(t)$ is uniformly bounded in time, with a bound $O(\lambda^{1-q})$ where $0\le q\le 1$.
 \item Consider the evolution equation obtained for $(B_t,D_t)$ by the standard 
``closure relation'', i.e., by setting $\kappa_6\to 0$ in (\ref{eq:inhomogGt}).
 Assume that this equation is stable under all uniformly bounded 
\defem{time-dependent} perturbations of the source term.
\end{enumerate}
If both of the above hold, then the solutions to the closure evolution equation
for $f_t$ remain close to $(B_t,D_t)$ uniformly in $t$.  In particular, the difference in the 
first (covariance) component is always $O(\lambda^{2-q})$.
This would validate the closure equations as good approximations even to 
$t\to\infty$ asymptotic behavior of the covariance function
for all sufficiently small $\lambda$.  Note that no scaling limit needs to be 
taken; in particular, it is \defem{not} claimed that the kinetic scaling limit of 
$B_t$ or $D_t$ would exist.

Another application can be obtained for 
a one-component case with $f_t=W_t^{\lambda}$ as follows: 
Consider the spatially homogeneous case and the exact
evolution equation for $W_t^{\lambda}$ obtained from (\ref{W_expansion}) by 
``dividing out'' $\delta(\mathbf{k}'+\mathbf{k})$ from both sides.
Take $R_t$ to include {all} terms which contain either $\kappa_4$ or 
$\kappa_6$.  Then the remaining pairing terms yield an explicit definition for $F_{s,t}$
such that (\ref{eq:startf}) holds.  (In fact, then $F_{s,t}[W]$ is equal to the right hand side of (\ref{final_sec_ord}).)
Next 
choose ``$\Phi_{s,t}$'' equal to the Boltzmann collision operator
$\lambda^2 \pazocal{C}$ with $\pazocal{C}$ defined in (\ref{eq:defCBP}).  If $W_t$ comes from an $\ell_1$-clustering 
state and the free evolution is sufficiently dispersive, 
then $F_{s,t}[W_s] = \pazocal{C}[W_s] + O(\lambda^2 (1+|t-s|)^{-p})$, 
with $p>1$.
(For instance, the estimates given in \cite{LuSp11}, in Proposition 7.4 and in the Appendix,  prove the bound with $p=3 d/7-1$ for
a nearest neighbor dispersion relation---the computation is essentially the same as in Appendix \ref{sec:nonGauss} below.
Using the notations defined in (\ref{eq:defell1clnorm}), the result also allows to quantify
the dependence on the $\ell_1$-clustering assumption: the bound is proportional to $\norm{\kappa_2(s)}^3_1 \lambda^2 (1+|t-s|)^{-p}$.
Hence, then $p>1$ at least if $d\ge 5$.
However, more careful estimates or the addition of next to nearest neighbor hopping could
improve the bound.)

Whenever this is the case, we obtain a second conjecture about the homogeneous 
DNLS equation.  Suppose that there is a metric for the cumulants such that the following results hold for
sufficiently regular $\omega$ and initial data:
\begin{enumerate}
 \item Suppose that the cumulants remain uniformly bounded in this metric, with an upper bound
 which implies that $\norm{\kappa_2}_1=O(1)$ and that all 
 higher order cumulants have $\ell_1$-clustering norm which is $O(\lambda^{-q})$ for some $0\le q<2$.
 \item Assume that the corresponding Boltzmann-Peierls equation is stable under 
all uniformly bounded \defem{time-dependent} perturbations of the source term.
\end{enumerate}
If both of the above hold, then the solutions $W_\tau$ to the Boltzmann-Peierls 
equation with initial data $W^{\lambda}_0$ are
$O(\lambda^{2-q})$ close to $W^{\lambda}_{\tau\lambda^{-2}}$ uniformly in $\tau$.    In particular, any stationary
limit $\lim_{t\to\infty} W^{\lambda}_t$ can differ from the limit of the solution to the Boltzmann-Peierls equation only by $O(\lambda^{2-q})$.

The main benefit from using the Boltzmann-Peierls equation instead of the 
closure hierarchy concerns the second assumption: the homogeneous 
Boltzmann-Peierls equations enjoy many simplifying properties and a priori estimates, see for 
instance \cite{spohn05}.  For example, they typically
have an entropy functional and an associated ``H-theorem'' which allow to 
classify all stationary solutions to the equation.  There are also many
techniques developed to control the convergence towards the stationary solution.

As a final example, let us remark that 
even when total uniformity in time cannot be achieved, it might be 
possible to go beyond the kinetic time-scales using the above methods.
Consider $F_{s,t}[W_t] = \pazocal{C}[W_t] + 
O(\lambda^2 (1+|t-s|)^{-p})$ for some $0<p<1$. Then the correction is not integrable and
the perturbation $\rho_t$ to the source term is $O(\lambda^2 t^{1-p})$.  This 
remains $O(\lambda^{\vep})$ for all $t=O(\lambda^{-(2-\vep)/(1-p)})$.
Hence, for instance, the earlier nearest neighbor estimate with $d=3$ would
imply that the corrections to the Boltzmann-Peierls equation
remain small for $t=O(\lambda^{-2-4/5+\delta})$, that is,
even for times much longer than the ones implied by the kinetic scaling limit.

\appendix

\section{Combinatorial definition of cumulants}
\label{app:cumulants}

It is shown in \cite{Smith1995} that cumulants are connected to moments 
via a formula which is very similar to the 
definition we used here for the Wick polynomials:  if $I\ne \emptyset$, for any $x\in I$ we have 
\begin{align}\label{eq:iterdefkappa}
 \E[y^{I}] = \sum_{E:x\in E\subset I} \E[y^{I\setminus E}] \kappa[y_E]\, .
\end{align}
(The formula follows straightforwardly from the identity $\partial_x \Mgen = 
\Mgen \partial_x\Cgen$.)
In fact, this formula allows a definition of cumulants which does not rely on 
differentiation or on the existence of exponential moments.
Namely, if $I_0\in \mathcal{I}$ is such that $\E[|y^I|]<\infty$ for all 
$I\subset I_0$, then to each $I\subset I_0$, $I\ne \emptyset$, we can
associate a number $\kappa[y_I]$ by requiring that
$\kappa[y_{I}]=\E[y^{I}]-\sum_{E:x\in E\subsetneq I} \E[y^{I\setminus E}] 
\kappa[y_E]$ with $x=(1,i_1)$.  The definition
is used inductively in $|I|\ge 1$ and it has a unique solution.  (Note that the empty cumulant
$\kappa[y_\emptyset]$ never appears in the moments-to-cumulants formula, and for our purposes 
it can be left undefined.  To be consistent with the derivatives of the generating function, we may for instance set $\kappa[y_\emptyset]:=0$.)
Therefore, these 
numbers have 
to coincide with the standard cumulants in the case when exponential moments 
exist and hence (\ref{eq:iterdefkappa}) holds.

The following known properties of cumulants can then be derived directly from the
above definition using induction in $|I|$ and applying techniques similar  to what
we have used for Wick polynomials in Section \ref{sec:Wickprop}:
\begin{enumerate}
 \item The cumulants are multilinear, in the same manner as was stated
  for Wick polynomials in Proposition \ref{eq:Wickmultlin}.
 \item The moments-to-cumulants expansion (\ref{eq:mtc}) holds.
 \item The cumulants are permutation invariant: if $I'$ is a permutation of
$I$, then $\kappa[y_{I'}]=\kappa[y_I]$.
 \item If joint exponential moments exist, then $\kappa[y_I]= \partial^I_\lambda\Cgen(0)$
 with $\Cgen(\lambda):=\ln \E[\rme^{\lambda\cdot y}]$.
\end{enumerate}
However, let us skip the proofs here.  In the text, we assume
these results to be known and refer to the references for details of their proofs.

\section{Classical particle system with random initial data}
\label{sec:App_Ham}

Consider the evolution of 
$N$ classical particles interacting 
via a polynomial interaction potential, with the initial data given by some random probability
measure.  We show here how it can be recast in the form of the evolution equation discussed 
in Section \ref{sec:cwdyn}.

We consider the random variables $y_j$, indexed by 
$J=\set{(i,n)}_{i,n}$, where $n$ is one of the $N$ different particle labels
and $i=1,2$ differentiates between the particle position and momentum:
we define $y_{(1,n)}(t) := q_n(t)$ and $y_{(2,n)}(t) := p_n(t)$.
If all particles move in $\R$, have the same mass, and have only pair 
interactions via the potential
$V(q) := \sum_{n',n; n'\ne n} \lambda_{nn'} \frac{1}{2 a}(q_n-q_{n'})^a$, $a\ge 
2$ even
and $\lambda_{n'n}=\lambda_{nn'}$, then we have
$\partial_{t}q_n(t)=p_n(t)$ and
\begin{align}
 \partial_{t}p_n(t) = -\sum_{n'\ne n} \lambda_{nn'} (q_n(t)-q_{n'}(t))^{a-1} \, 
.
\end{align}
Here $(q_n-q_{n'})^{a-1} = \sum_{k=0}^{a-1} (-1)^{a-1-k}\binom{a-1}{k} q_n^k 
q_{n'}^{a-1-k}$, and if we define $I_{n,n',k}$ as a sequence of length $a-1$
containing first $k$ repetitions of $(1,n)$ and then $a-1-k$ repetitions of 
$(1,n')$, then by (\ref{inv_wick}) we have
$q_n(t)^k q_{n'}(t)^{a-1-k}=\sum_{V\subset I_{n,n',k}} \w{y(t)^V} 
\E(y(t)^{I_{n,n',k} \backslash V})$.  Define thus as the collection
$\mathcal{I}_{(2,n)}$ all such sequences $U$ which contain $k_1$ repetitions of 
$(1,n)$ followed by $k_2$ repetitions of $(1,n')$
where $n'\ne n$, $k_1,k_2\ge 0$, and $k_1+k_2\le a-1$.  Set also for each $U\in 
\mathcal{I}_{(2,n)}$
\begin{align}
 M^U_{(2,n)}(t) := \lambda_{nn'}\sum_{k=k_1}^{a-1-k_2} \binom{a-1}{k}(-1)^{a-k}
 \sum_{V\subset I_{n,n',k}} \E(y(t)^{I_{n,n',k} \backslash V}) \cf(V=U)\, .
\end{align}
Therefore, (\ref{eq:basicev}) holds for all $j\in J$ after we also define 
$\mathcal{I}_{(1,n)}:=\set{\emptyset,((2,n))}$ and set
$M^\emptyset_{(1,n)}(t) := \E(y_{(2,n)}(t))$ and $M^U_{(1,n)}(t) := 1$ if 
$U=((2,n))$.

\section{Estimation of the first order non-pairing contributions to 
(\ref{W_expansion})}\label{sec:nonGauss}

In this appendix, we show how to estimate the first order non-pairing 
contraction terms in (\ref{W_expansion}). 
To this end we need to make an assumption on the dispersion relation $\omega(\mathbf{k})$.
Let us consider the so called free propagator
\begin{align}\label{eq:defpt}
p_t(\mathbf{x})=
\int_{\mathbb{T}^d} d \mathbf{k} 
 \rme^{\ci 2 \pi \mathbf{x} \cdot \mathbf{k} }
\rme^{-\ci t \omega(\mathbf{k})} \,.
\end{align}
As in the assumption ``(DR2)'' in \cite{LuSp11},
we now suppose that there are $ C, \delta >0$ such that for all $ t \in \mathbb{R}$,
\begin{align}\label{assumption_propagator}
\norm{p_t}_3^3 = \sum_{\mathbf{x} \in \mathbb{Z}^d} |p_t(\mathbf{x})|^3 \leq C 
(1+t^2)^{-(1+\delta)/2} \,.
\end{align}
Furthermore, we assume also the already mentioned $ \ell_1$-clustering property 
(see section \ref{sec:cwdyn}) which we slightly rephrase as follows for each 
cumulant of order $n$: we require that 
\begin{align}\label{eq:defell1clnorm}
\norm{\kappa_n}_1 :=
\sup_{\sigma \in \{\pm 1 \}^n}\sum_{x \in (\Z^{d})^n} \cf(\mathbf{x}_1=0) 
\big\vert \kappa(\psi(\mathbf{x}_1, \sigma_1), \ldots , \psi(\mathbf{x}_n, 
\sigma_n)) \big\vert < \infty \,.
\end{align}
We recall that the physical meaning of this condition is that the cumulants 
decay fast enough in space so that they are summable, once the translational 
invariance is taken into account.

We recall that the first order non-pairing contributions in (\ref{W_expansion}) 
are 
\begin{align}\label{first_order}
 &  
-\ci\lambda \sigma \int_0^t ds \int_{(\T^{d})^3} d\mathbf{k}_1 
d\mathbf{k}_2 d\mathbf{k}_3 \delta(\mathbf{k} - \mathbf{k}_1 - \mathbf{k}_2 
-\mathbf{k}_3)\rme^{\ci s(\sigma \omega +\omega_1 - \sigma \omega_2 - 
\omega_3)}
\nonumber \\ & \quad \times
\kappa[ a(\mathbf{k}_1,-1); a(\mathbf{k}_2,\sigma); a(\mathbf{k}_3,1); a(\mathbf{k}', \sigma') ] 
\end{align}
and a term which is obtained from (\ref{first_order}) by swapping $(\mathbf{k},\sigma) \leftrightarrow (\mathbf{k}',\sigma')$.
As stated in (\ref{trans_inv_cum}), by translation invariance we have 
\begin{align}
\kappa[ a(\mathbf{k}_1,-1); a(\mathbf{k}_2,\sigma); a(\mathbf{k}_3,1); 
a(\mathbf{k}', \sigma') ] = 
\delta(\mathbf{k}_1+\mathbf{k}_2+\mathbf{k}_3+\mathbf{k}') 
\widehat{F}(\mathbf{k}_1,\mathbf{k}_2,\mathbf{k}_3,\mathbf{k}',\sigma, \sigma')
\end{align}
where $F(\mathbf{x}_1, \mathbf{x}_2, \mathbf{x}_3,\mathbf{x}_4, \sigma, \sigma' 
)=\cf(\mathbf{x}_1=0) \kappa[ a(\mathbf{x}_1,-1) ; a(\mathbf{x}_2,\sigma);
a(\mathbf{x}_3,1); a(\mathbf{x}_4, \sigma') ]$.  Clearly, $\norm{F}_1 \le \norm{\kappa_4}_1<\infty$ by the assumed $\ell_1$-clustering.

Therefore, the term in (\ref{first_order}) is bounded by
\begin{align}
& \lambda \bigg\vert \int_0^{t} ds \int_{(\T^{d})^3} 
d\mathbf{k}_1 d\mathbf{k}_2 d\mathbf{k}_3 \delta(\mathbf{k} - \mathbf{k}_1 - 
\mathbf{k}_2 -\mathbf{k}_3)\rme^{\ci s(\sigma \omega +\omega_1 - \sigma \omega_2 - \omega_3)}
 \nonumber \\ & \qquad \times 
\delta(\mathbf{k}_1+\mathbf{k}_2+\mathbf{k}_3+\mathbf{k}') 
\widehat{F}(\mathbf{k}_1,\mathbf{k}_2,\mathbf{k}_3,\mathbf{k}',\sigma, \sigma') 
\bigg\vert \nonumber \\ & \quad 
\le \lambda \delta(\mathbf{k}+\mathbf{k}') \int_0^{t} ds 
\bigg\vert \int_{(\T^{d})^2} d\mathbf{k}_1 d\mathbf{k}_2 \left.
\rme^{\ci s(\omega_1 - \sigma \omega_2 - \omega_3)}
\widehat{F}(\mathbf{k}_1,\mathbf{k}_2,\mathbf{k}_3,\mathbf{k}',\sigma, \sigma')\right|_{\mathbf{k}_3=\mathbf{k}-\mathbf{k}_1-\mathbf{k}_2}
\bigg\vert \nonumber \\ & \quad
\le\lambda  \delta(\mathbf{k}+\mathbf{k}') \int_0^{t} ds 
 \sum_{\mathbf{x}_1, \mathbf{x}_2, \mathbf{x}_3,\mathbf{x}_4}
|F(\mathbf{x}_1, \mathbf{x}_2, \mathbf{x}_3,\mathbf{x}_4, \sigma, \sigma')|
\nonumber \\ & \qquad \times 
\bigg\vert \sum_{\mathbf{y}}\rme^{-\ci 2 \pi \mathbf{k} \cdot (\mathbf{y}-\mathbf{x}_3)}
p_{-s}(\mathbf{y})p_{\sigma s}(\mathbf{y}-\mathbf{x}_2)p_{s}(\mathbf{y}-\mathbf{x}_3) \bigg\vert \nonumber \\ & \quad
 \leq \lambda \delta(\mathbf{k}+\mathbf{k}') \norm{\kappa_4}_1 \int_0^{t} ds\, \norm{p_s}_3^3 
 \leq \lambda C \delta(\mathbf{k}+\mathbf{k}') 
 \norm{\kappa_4}_1 \int_0^{t} ds(1+s^2)^{-(1+\delta)/2} 
\nonumber \\ & \quad
\le \lambda C' \delta(\mathbf{k}+\mathbf{k}') 
\norm{\kappa_4}_1 
\end{align}
where $C'$ is a constant which depends only on $C$ and $\delta$.  We have used the inverse Fourier transform of
$\FT{F}$ and (\ref{eq:defpt}) in the second inequality and H\"older's inequality in the third one.
Since the bound in invariant under the swap $(\mathbf{k},\sigma) \leftrightarrow (\mathbf{k}',\sigma')$, 
it bounds also the second non-pairing contribution in (\ref{W_expansion}).
Therefore, we see that the first order contributions are $O(\lambda)$ uniformly in $t$.

\end{document}